\documentclass[final,twocolumn]{IEEEtran}

\usepackage{graphicx}
\usepackage{amssymb}
 
\usepackage{amsmath}
\usepackage{dsfont}
\usepackage{bm}
\usepackage{amsthm}
\usepackage{subfigure}

\def\R{{\mathds R}}
\def\C{{\mathds C}}

\def\e{\mathrm{e}}
\def\E{\mathrm{E}}
\def\VAR{\mathrm{VAR}}

\newcommand{\be}{\begin{equation}}
\newcommand{\ee}{\end{equation}}

\newcommand{\bzero}{{\mbox{\boldmath $0$}}}
\newcommand{\bI}{{\mbox{\boldmath $I$}}}
\newcommand{\bz}{{\mbox{\boldmath $z$}}}
\newcommand{\bn}{{\mbox{\boldmath $n$}}}
\newcommand{\bv}{{\mbox{\boldmath $v$}}}
\newcommand{\bp}{{\mbox{\boldmath $p$}}}

\newcommand{\bC}{{\mbox{\boldmath $C$}}}

\newcommand{\bZ}{{\mbox{\boldmath $Z$}}}
\newcommand{\bR}{{\mbox{\boldmath $R$}}}

\newcommand{\bS}{{\mbox{\boldmath $S$}}}
\newcommand{\bP}{{\mbox{\boldmath $P$}}}

\newcommand{\Pvtildeperp}{{\mbox{\boldmath $\bP^\perp_{\tilde{\bv}}$}}}

\newtheorem{proposition}{Proposition}

\newtheorem{insight}{Insight}

\newcommand{\test}{\mbox{$
\begin{array}{c}
\stackrel{ \stackrel{\textstyle H_1}{\textstyle >} }{ 
\stackrel{\textstyle <}{ \textstyle  H_0} }

\end{array}
$}}

\usepackage{color}

\begin{document}

\title{
CFAR Feature Plane: a Novel Framework for the Analysis and Design of Radar Detectors
}

\author{Angelo Coluccia, \IEEEmembership{Senior Member, IEEE}, Alessio Fascista, \IEEEmembership{Member, IEEE}, and Giuseppe Ricci, \IEEEmembership{Senior Member, IEEE}
\thanks{All authors are with the Dipartimento di Ingegneria dell'Innovazione,
        Universit\`a del Salento, Via Monteroni, 73100 Lecce, Italy.
        E-Mail: name.surname@unisalento.it}
}

\maketitle

\begin{abstract}
Since Kelly's pioneering work on GLRT-based adaptive detection, many solutions  have been proposed to enhance either  selectivity or  robustness of radar detectors to mismatched signals. In this paper such a problem is addressed in a different space, called CFAR feature plane  and given by a suitable maximal invariant, where observed data are mapped to clusters that can be analytically described. The characterization of the trajectories and shapes of such clusters is provided and exploited for both analysis and design purposes, also shedding new light on the behavior of several well-known detectors. Novel linear and non-linear detectors are proposed with diversified robust or selective behaviors, showing that through the proposed framework it is not only possible to achieve the same performance of well-known receivers obtained by a radically different design approach (namely GLRT), but also to devise detectors with unprecedented behaviors: in particular, our results show that the highest standard of selectivity can be achieved without sacrifying neither detection power under matched conditions nor CFAR property.
\end{abstract}


\section{Introduction}\label{sec:introduction}

The well-known  problem of detecting the possible presence of a (point-like) target's coherent return from a given cell under test (CUT) in range, doppler, and azimuth, is classically formulated as the following binary hypothesis testing problem:
\begin{equation}
\left\{
\begin{array}{ll}
H_{0}: & \bz =  \bn \\
H_{1}: & \bz = \alpha \bv + \bn
\end{array} 
\right. \label{eq:binary_test}
\end{equation}
where
$\bz \in \C^{N \times 1}$, $\bn \in \C^{N \times 1}$, and $\bv \in \C^{N \times 1}$ denote 
the received vector, the corresponding noise term, and 
the known space-time steering vector of the useful target echo, respectively, while $\alpha \in \C$ is an unknown deterministic parameter depending on radar cross-section, multipath, and other channel effects.
In general $N$ is the number of processed samples from the CUT; it might be the number of antenna array elements times the number of pulses \cite{Ward,Klemm-STAP}.

Kelly \cite{Kelly} derived a generalized likelihood ratio test (GLRT) for problem \eqref{eq:binary_test} assuming complex normal distributed noise with zero mean and unknown (Hermitian) positive definite covariance matrix $\bC$, denoted by ${\cal CN}_N (\bzero, \bC)$, and $K \geq N$ independent and identically distributed training (or secondary) data $\bz_1, \ldots, \bz_{K}$ (independent of $\bz$, free of target echoes, and  sharing with the CUT the statistical characteristics of the noise).
Remarkably, the resulting detection statistic does not depend upon any unknown parameter under the $H_0$ hypothesis, hence Kelly's detector 
 has the constant false alarm rate (CFAR) property, which is very important in practice to ensure that the detection threshold can be set to achieve a chosen probability of false alarm ($P_{fa}$) irrespective of the clutter statistics. %
 In  \cite{Kelly89} the performance of Kelly's detector is assessed when the actual
steering vector is not aligned with the nominal one.
The analysis shows that it is a selective receiver, i.e., it tends to reject signals arriving from directions  different from the nominal one. A selective detector is desirable for target localization; instead, a certain level of robustness to mismatches is preferable when the radar is working in searching mode. A prominent example of robust receiver is the adaptive matched filter (AMF) derived by Robey et al. in \cite{Kelly-Nitzberg} as two-step GLRT, i.e., by first assuming that the covariance matrix is known and then replacing it by the sample covariance based on secondary data in the detection statistic.

Following such pioneering papers, many  authors have addressed the problem of enhancing either the selectivity or the robustness to  mismatches, while ensuring the CFAR property.
Typical design procedures include statistical tests with modified hypotheses, asymptotic arguments, approximations, and ad-hoc strategies, which have led to a plethora of different detectors (several of them will be analyzed in Sec. \ref{sec:receivers}). 

In this paper we have a new look at the problem of designing and analyzing CFAR detectors for \eqref{eq:binary_test} with desired robust or selective behavior. The idea is to visualize observation data as point clouds forming \emph{clusters} in a suitable plane, with the aim of separating the $H_0$ cluster from the union of the matched and mismatched $H_1$ clusters, in case robustness is of interest, or   the union of the $H_0$ and mismatched-$H_1$ clusters from the matched-$H_1$ cluster, in case selectivity is desired.
We will show that the decision region boundary in such a space, hereafter referred to as CFAR feature plane (CFAR-FP), can be easily derived and interpreted as a linear or non-linear classifier, shedding new light on the behavior of several well-known detectors. Besides, a main contribution of the paper is the analytical characterization of the trajectories and shapes of the clusters parameterized by the signal-to-noise ratio (SNR) and mismatch on the steering vector, which turns out to be very useful for both analysis and design purposes.
As a second contribution we provide different criteria for obtaining detectors with desired behaviors. In particular, we discuss a tunable detector based on the likelihood ratio test in the CFAR-FP, showing that it behaves as a robust or selective receiver based on the setting of a design parameter, and outperforms classical competitors; moreover, we discuss how to design novel linear and non-linear detectors aimed at promoting either robustness or selectivity, while keeping satisfactory performance under matched conditions.

Our rationale is rooted in the theory of invariant detection, which puts forward the  idea of focusing on detectors that are invariant under transformations that leave the decision problem unchanged. %
In \cite{Bose} the group of such admissible transformations  is characterized for problem \eqref{eq:binary_test} and it is proved that the corresponding invariant tests can be expressed in terms of a \emph{maximal invariant} given by any pair of statistics with 1-1 relation to Kelly's and AMF statistics; thus, invariant tests also possess the CFAR property. 
The statistical characterization of the maximal invariant has been used over the years to derive $P_{fa}$ and probability of detection ($P_d$) formulae for several detectors, also under mismatched conditions. In \cite{Pulsone-Rader} a thorough way to inspect the behavior of invariant detectors is identified in the contours of constant $P_d$ in the SNR-$\cos^2 \theta$ plane (so-called mesa plots), since their performances only depend  on the SNR and cosine squared of the angle $\theta$ between the nominal and mismatched steering vectors.
In \cite{Bose2} a theoretical bound on the achievable performance of invariant detectors (in terms of $P_d$ under matched conditions)  is obtained by considering the (non-implementable) maximum power invariant (MPI) test, analyzed also in a plane given by the marginal cumulative distribution functions (CDFs) of a maximal invariant. In the CFAR-FP framework we propose here, the problem is instead addressed in a plane given by the maximal invariant directly, not marginal CDFs of maximal invariants (as in \cite{Bose2}) nor performance parameters SNR-$\cos^2 \theta$ (as in mesa plots).  This also enables a more intuitive interpretation under matched/mismatched conditions using general concepts adopted e.g. in machine learning based classification (data clusters,  decision region boundary, and linear/non-linear classifier); as a result, new insights about a number of well-known receivers are provided, exploiting the analytical characterization of the trajectories and shapes of the clusters (parameterized by SNR and $\cos^2 \theta$) we derive in the paper.
On the design side, one of the detectors we propose is an ad-hoc tunable version of the MPI test in \cite{Bose2}; in particular, we will show that while in \cite{Bose2} robustness/selectivity aspects  are not considered, our solution can exhibit a robust or selective behavior with satisfactory performance under matched conditions.
Additionally, we discuss different criteria to design a radar detector in the CFAR-FP with desired behavior; this is a new way of thinking, quite different from the traditional hypothesis testing approach, and again leverages the characterization of the cluster trajectories we provide in this work. Through the proposed framework it is possible to achieve the same performance of some well-known receivers obtained by a radically different design approach (namely GLRT) but also remarkable unprecedented behaviors: in particular, ours is the first design tool able to deliver a CFAR detector with very strong selectivity but no $P_d$ loss under matched conditions.

The rest of the paper is organized as follows. Sec. \ref{sec:receivers} is devoted to the analysis of several well-known detectors in the CFAR-FP and to the proposed tunable detector. In Sec. \ref{sec:characterization} the statistical characterization of clusters in the CFAR-FP is developed, while Sec. \ref{sec:perf_ass} focuses on the design and performance assessment. We conclude in Section \ref{sec:conclusions}.

\section{Novel analysis framework for CFAR detectors}\label{sec:receivers}

\begin{table*}
\caption{Classification region boundary in the CFAR-FP ($\beta$-$\tilde{t}$ plane)}
\begin{center}
\begin{tabular}{c|c}
\hline\hline\\[-0.2cm]
Name & Curve equation (either explicit $\tilde{t}=f(\beta)$ or implicit $f(\beta,\tilde{t})=0$)  \\[0.1cm]
\hline\\[-0.2cm]
Kelly's detector \cite{Kelly} & $\tilde{t} = \eta_\text{\tiny K}$\\[0.2cm]
adaptive matched filter (AMF) \cite{Kelly-Nitzberg} & $\tilde{t} =  \eta_\text{\tiny AMF} \beta$\\[0.2cm]
adaptive coherence detector (ACE)  \cite{ACE} & $\tilde{t} =  - \frac{\eta_\text{\tiny ACE}}{1-\eta_\text{\tiny ACE}} \beta + \frac{\eta_\text{\tiny ACE}}{1-\eta_\text{\tiny ACE}}$\\[0.2cm]
energy detector (ED) & $\tilde{t} =  (\eta_\text{\tiny ED}+1) \beta -1$\\[0.2cm]
Kalson's detector \cite{Kalson} & $\tilde{t} =  \frac{ (1-\epsilon_\text{\tiny Kalson})\eta_\text{\tiny Kalson}}{1-\epsilon_\text{\tiny Kalson} \eta_\text{\tiny Kalson}} \beta + \frac{\epsilon_\text{\tiny Kalson} \eta_\text{\tiny Kalson}}{1-\epsilon_\text{\tiny Kalson} \eta_\text{\tiny Kalson}}$ \,\, $(0 \leq \epsilon_\text{\tiny Kalson} \leq 1)$ \\[0.2cm]
adaptive beamformer orthogonal rejection test (ABORT)  \cite{Pulsone-Rader} & $\tilde{t} =  - \beta +\frac{\eta_\text{\tiny A}}{1-\eta_\text{\tiny A}}$\\[0.2cm]
\hline\\[-0.1cm]
whitened-ABORT detector (WABORT) \cite{W-ABORT} & $\tilde{t} = \frac{\eta_\text{\tiny WA}}{\beta} - 1$\\[0.2cm]
parametric Kelly-WABORT detector (KWA) \cite{KWA} & $\tilde{t} = \frac{\eta_\text{\tiny KWA}}{\beta^{2\epsilon_\text{\tiny KWA}-1}} - 1$ \,\, $(\epsilon_\text{\tiny KWA} >0)$ \\[0.2cm]
Rao's test \cite{Rao} & $\tilde{t} = \frac{\eta_\text{\tiny Rao}}{\beta-\eta_\text{\tiny Rao}} $\\[0.2cm]
\hline\\[-0.1cm]
conic acceptance detector (CAD) \cite{BOR-MorganClaypool}  & $\frac{1-\beta+\tilde{t}}{\beta} \! - \! \frac{1}{1+\epsilon_\text{\tiny CAD}^2} \left[ \sqrt{\frac{1}{\beta}-1} -\epsilon_\text{\tiny CAD} \sqrt{\frac{\tilde{t}}{\beta}} \right]^2 \!\! u\left( \frac{1-\beta+\tilde{t}}{\beta} -\frac{\tilde{t}}{\beta} (1+\epsilon_\text{\tiny CAD}^2) \right) - \eta_\text{\tiny CAD} =0$  \\
& $(\epsilon_\text{\tiny CAD} >0)$, where $u(x)=\left\{ \begin{array}{ll} 1 & x\geq 0\\ 0 & x<0 \end{array} \right.$ is the Heaviside step function\\[-0.3cm]
\\[0.2cm]
conic acceptance-rejection detector (CARD)  \cite{Bandiera-DeMaio-Ricci-CONI,BOR-MorganClaypool} & $\left[ \epsilon_\text{\tiny CARD}\sqrt{\frac{\tilde{t}}{\beta}} -  \sqrt{\frac{1}{\beta}-1} \right]^2 \mathrm{sgn} \left( \epsilon_\text{\tiny CARD} \sqrt{\frac{\tilde{t}}{\beta}}  -  \sqrt{\frac{1}{\beta}-1} \right) - \eta_\text{\tiny CARD} =0$  \\
& $ (\epsilon_\text{\tiny CARD} >0)$, where $ \mathrm{sgn}(x) =\left\{ \begin{array}{ll} +1 & x\geq 0\\ -1 & x<0 \end{array} \right.$ is the sign function\\[0.4cm]
random-signal robustified detector (ROB) \cite{arxivROB} & $\left\{
\begin{array}{ll}
\tilde{t}
 =  \frac{\eta_\text{\tiny ROB}}{ 1-\frac{1}{\zeta_{\epsilon}} } \beta \left[ \left(\zeta_{\epsilon} -1 \right)
\left( \frac{1}{\beta}-1\right) \right]^{\frac{1}{\zeta_{\epsilon}}} -1, &
\beta \in (0, 1-\frac{1}{\zeta_{\epsilon}}] \\[0.2cm]
\tilde{t} = \eta_\text{\tiny ROB} -1, & \beta \in [1-\frac{1}{\zeta_{\epsilon}},1)
\end{array}
\right. 
$\\
& where $\zeta_{\epsilon}=\frac{K+1}{N}(1+\epsilon_\text{\tiny ROB})$ with $\epsilon_\text{\tiny ROB} \geq 0$\\[0.3cm]
\hline\\[-0.1cm]
natural parametric detector  (NAT) in Sec. \ref{sec:NAT} & $\mathrm{e}^{-\frac{\epsilon_\text{\tiny NAT} \beta}{1+\tilde{t}}}  \sum_{h=0}^{K-N+1} a_h(\epsilon_\text{\tiny NAT}) \left( \frac{\beta \tilde{t}}{1+\tilde{t}} \right)^h - \eta_\text{\tiny NAT} = 0$\\[0.2cm]
& where $a_h(\epsilon_\text{\tiny NAT})=\frac{\binom{K-N+1}{h}}{h!} \epsilon_\text{\tiny NAT}^h$ with $\epsilon_\text{\tiny NAT} \geq 0$ \\[0.3cm]
\hline\hline
\end{tabular}
\end{center}
\label{tabella}
\end{table*}%

\subsection{CFAR processing chain and feature plane}

As mentioned, over the past decades different ideas have been considered, based on statistical tests with modified hypotheses, asymptotic arguments, approximations, and ad-hoc strategies, to promote robustness or selectivity in radar detectors.
Proposition \ref{teo:FS} introduces the proposed CFAR-FP framework, in which invariant detectors are seen as learning machines that guarantee the CFAR property through a suitable processing chain. The equations describing their corresponding decision region boundary are obtained for several well-known detectors, showing that they have an intuitive interpretation in the CFAR-FP as linear or non-linear  classifiers, so shedding new light on their robust/selective behaviors and trade-offs.

\begin{proposition}\label{teo:FS}
Define the mapping chain
\begin{equation}
\{\bz, \bZ =[\bz_1 \cdots \bz_K]\} \mapsto (\bz, \bS^{-1}) \mapsto (s_1,s_2) \mapsto (\beta, \tilde{t}) \label{eq:FS_mapping}
\end{equation}
where $(s_1,s_2)=(\bz^{\dagger} \bS^{-1} \bz, \frac{|\bz^{\dagger} \bS^{-1} \bv|^{2}}{\bv^{\dagger} \bS^{-1} \bv})$, $\bS=\bZ\bZ^\dag$ denotes the scatter matrix (i.e., $K$ times the sample covariance), and
\be
(\beta, \tilde{t}) = \left(\frac{1}{1+ s_1  - s_2}, \frac{s_2 }{1+ s_1 - s_2} \right) \in (0,1)\times \R_+
\label{eq:FSpair}
\ee
with $\R_+ = (0, +\infty)$ and ${}^\dag$ conjugate transpose (Hermitian). All detectors built through the processing chain \eqref{eq:FS_mapping} are linear or non-linear classifiers in the $\beta$-$\tilde{t}$ CFAR-FP and possess the CFAR property; 
the decision region boundary equations of several well-known detectors are summarized in Table \ref{tabella}.
\end{proposition}

\begin{proof}
The first part of the statement is just a reparametrization of the relationship between the maximal invariants $(t_\text{\tiny Kelly},t_\text{\tiny AMF})$ \cite{Bose} and $(\beta, \tilde{t})$, the latter being the statistical representation often used to derive $P_{fa}$ and $P_d$ formulae for several well-known detectors \cite{BOR-MorganClaypool}, hence the CFAR property follows. In particular, $\tilde{t} = \frac{t_\text{\tiny Kelly}}{1-t_\text{\tiny Kelly}}$ where $t_{\text{\tiny{Kelly}}}$ is Kelly's statistics (ref. eq.  \eqref{eq:Kelly}), and $\beta= (1+\bz^{\dagger} \bS^{-1} \bz  - \frac{|\bz^{\dagger} \bS^{-1} \bv|^{2}}{\bv^{\dagger} \bS^{-1} \bv})^{-1}$. The second part of the statement is obtained, for each given detector $X$, by  rewriting its detection statistic $t_X$ in terms of $\beta$ and $\tilde{t}$ only, and then studying the equation $t_X(\beta, \tilde{t}) = \eta_X$, where $\eta_X$ is the threshold for a chosen  $P_{fa}$; $t_X$ may also depend on a tunable parameter $\epsilon_X$. In many cases this equation can be solved   explicitly in  $\tilde{t}$ as a linear or non-linear function of $\beta$, ref. Table \ref{tabella} (the calculations are simple hence omitted, but some details can be found in the analysis below).
\end{proof}

Before discussing insights resulting from Proposition \ref{teo:FS}, some comments are in order. 
First, the choice $(\beta, \tilde{t})$ is mathematically convenient, since the two variables are independent under $H_0$, with well-known distributions \cite{Kelly,Kelly89,BOR-MorganClaypool}, and one of them is equivalent to Kelly's detector that is considered the benchmark for matched conditions (more on this point later).
Other maximal invariants might be adopted, e.g. $(t_\text{\tiny Kelly},t_\text{\tiny AMF})$, $(s_1,s_2) = (t_\text{\tiny ED},t_\text{\tiny AMF})$ (ref. eqs. \eqref{eq:ED} and \eqref{eq:AMF}), or the one in \cite{Bose2}; reasons for our choice will be highlighted in the sequel.

Second, the mapping chain \eqref{eq:FS_mapping} can be interpreted, in the spirit of \cite{Haykin},  as  layers of a learning machine\footnote{In general a learning machine is a weighted interconnection of processing units, e.g. a neural network, with many degrees of freedom to be determined.} with a peculiar structure. This provides a general multi-layer scheme for a processing chain that guarantees the CFAR property, as depicted in Fig. \ref{fig:network}. After the input layer with the raw data $\{\bz, \bz_1, \ldots, \bz_K\}$, the first two hidden layers serve for what is called in \cite{Haykin} an \emph{invariant feature extractor}, i.e., the input data undergo two stages of compression. 
This processing can be interpreted as follows: given the constraint of CFAR property guarantee, the first hidden layer rules out any other uses of the secondary data different from the construction of the sample matrix inverse $\bS^{-1}$, while the second hidden layer rules out any data compression function different from a transformation that depends on the  energy (norm) of the received signal vector ($s_1$) and its scalar product with the energy-normalized steering vector ($s_2$). 
The third hidden layer is instead  an invertible mapping to obtain a more convenient mathematical representation in terms of the $(\beta, \tilde{t})$ variables. The latter capture invariant properties of the data and, on the other hand, help to restrict the architecture of the learning machine; accordingly, here they will be referred to  as  \emph{features}, as customary in the machine learning terminology.  

\begin{figure}
\centering
\includegraphics[width=7.5cm]{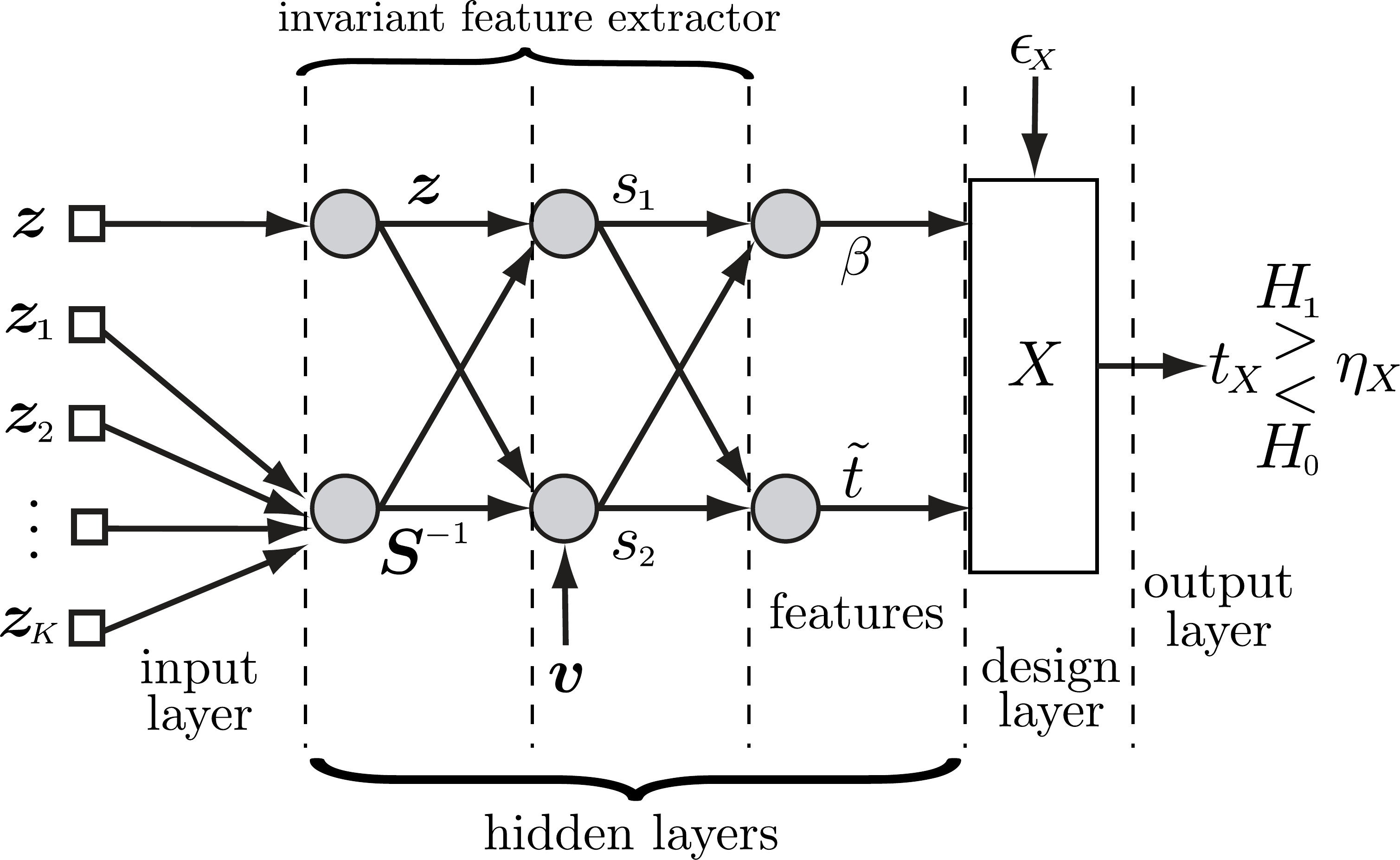}
\caption{Layered representation, from raw data to feature space $(\beta,\tilde{t})$, of the CFAR processing chain of a detector $X$ with statistic $t_X$ and threshold $\eta_X$.} \label{fig:network}
\end{figure}

What determines the peculiar behavior of different detectors is the fourth layer, where all the design choices take effect and the actual ``learning'' is performed. 
Typically, the design/learning of the decision function in the fourth layer is unsupervised: indeed, hypothesis testing tools such as GLRT are the most usual in the traditional approach. More recently, cognitive and machine (even deep) learning techniques are providing additional (often supervised, i.e., data-driven) ways to design radar detectors \cite{deep_radar,Gini1,Gini2,arxivKNN}. Both such approaches are subsumed in the fourth layer, hence can be treated in a unified way under the proposed framework.

\subsection{Interpretation of well-known receivers in the CFAR-FP}\label{sec:interpretation}

Figs. \ref{fig:robustiN16K32}-\ref{fig:selettiviN16K32} show the equations of Table \ref{tabella} (discussed one-by-one below) in the CFAR-FP,  superimposed to point clusters obtained by generating the random variable $\bz$ under three different conditions --- $H_0$ (blue dots), $H_1$ under matched conditions (red dots), $H_1$ under mismatched conditions (cyan dots) --- and  $\bZ$ as usual under the null hypothesis, then mapping the raw data to the  $\beta$-$\tilde{t}$ plane. 
The target amplitude  $\alpha$ is generated deterministically according to the SNR
\be
\gamma = |\alpha|^2 \bp^{\dagger}  \bC^{-1} \bp \in\R_+
\label{eq:SNR}
\ee
where $\bp$ is a  steering vector whose mismatch with respect to the nominal steering $\bv$ can be quantified in terms of the cosine squared of the angle $\theta$ between $\bp$ and $\bv$ in the whitened space:
\be
\cos^2\theta=\frac{|\bp^\dag \bC^{-1}\bv|^2}{\bv^\dag\bC^{-1}\bv\, \bp^\dag
\bC^{-1}\bp} \in [0,1].
\label{eqn:defcostheta}
\ee
 Under matched conditions, of course, $\bp=\bv$ and $\cos^2 \theta=1$.

It is well-known \cite{Kelly,Kelly89,BOR-MorganClaypool} that the  performance of invariant detectors in terms of $P_d$ (for chosen $P_{fa}$) depend only on the SNR $\gamma$ \eqref{eq:SNR} and mismatch $\cos^2 \theta$ \eqref{eqn:defcostheta}; in the CFAR-FP this means that, once the desired  $P_{fa}$ has been set, all thresholds can be computed accordingly, enabling a fair comparison. Thus, since thresholds become fixed (given $P_{fa}$) the only relevant parameters are $\gamma$ and $\cos^2 \theta$ (irrespective of the values of $\bv$ and  $\bC$ used to generate the data).\footnote{For reproducibility, we specify that in Figs. \ref{fig:robustiN16K32}-\ref{fig:selettiviN16K32} clusters have 5000 points each, and we assume $\bv=[1\  \e^{i2\pi f_d} \ \cdots \ \e^{i2\pi (N-1)f_d} ]^{T}$, $N=16$, $K=32$, normalized Doppler frequency $f_d=0.08$ (a small value such that the target competes with low pass clutter), $\bp$ defined as $\bv$ but with $f_d+\delta_f$ and $\delta_f = 0.3/N$ ($\cos^2 \theta = 0.65$), $\gamma = 15$ dB, $P_{fa}=10^{-4}$, and as $\bC$ the sum of  a Gaussian-shaped clutter  and white (thermal) noise 10 dB weaker, i.e.,
$\bC = \bR_c + \sigma_n^2 \bI_N$ with  the $(m_1,m_2)$th element of the matrix $\bR_c$ given by
$[\bR_c]_{m_1,m_2} \propto \exp\{- 2\pi^2\sigma_f^2(m_1-m_2)^2\}$ and $\sigma_f \approx 0.051$ (corresponding to a one-lag correlation coefficient of the clutter component equal to $0.95$).}

To start with, we recall that Kelly's one-step GLRT  \cite{Kelly}
\begin{align}
t_{\text{\tiny{Kelly}}}  &= 
 \frac{|\bz^{\dagger} \bS^{-1} \bv |^2}{\bv^{\dagger} \bS^{-1} \bv \, (1+ \bz^{\dagger} \bS^{-1} \bz )} = \frac{\tilde{t}}{1+\tilde{t}} \label{eq:Kelly}
\end{align}
is a selective detector;  conversely, the adaptive matched filter (AMF) \cite{Kelly-Nitzberg} obtained through a two-step GLRT procedure \begin{align}
t_{\text{\tiny{AMF}}}  &= 
 \frac{|\bz^{\dagger} \bS^{-1} \bv |^2}{\bv^{\dagger} \bS^{-1} \bv } = \frac{\tilde{t}}{\beta} \label{eq:AMF}
\end{align} 
is a robust receiver.
As mentioned, $(t_{\text{\tiny{Kelly}}}, t_{\text{\tiny{AMF}}})$ is the maximal invariant for the problem at hand, hence is not surprising that the only statistics that need to be retained from the raw data ($s_1$ and $s_2$) are building blocks of \eqref{eq:Kelly}-\eqref{eq:AMF}; the only difference is that $s_1$ and $s_2$ have a  direct engineering interpretation as energy and correlation (scalar product, i.e., matched filtering).

\begin{figure}
\centering
\includegraphics[width=8cm]{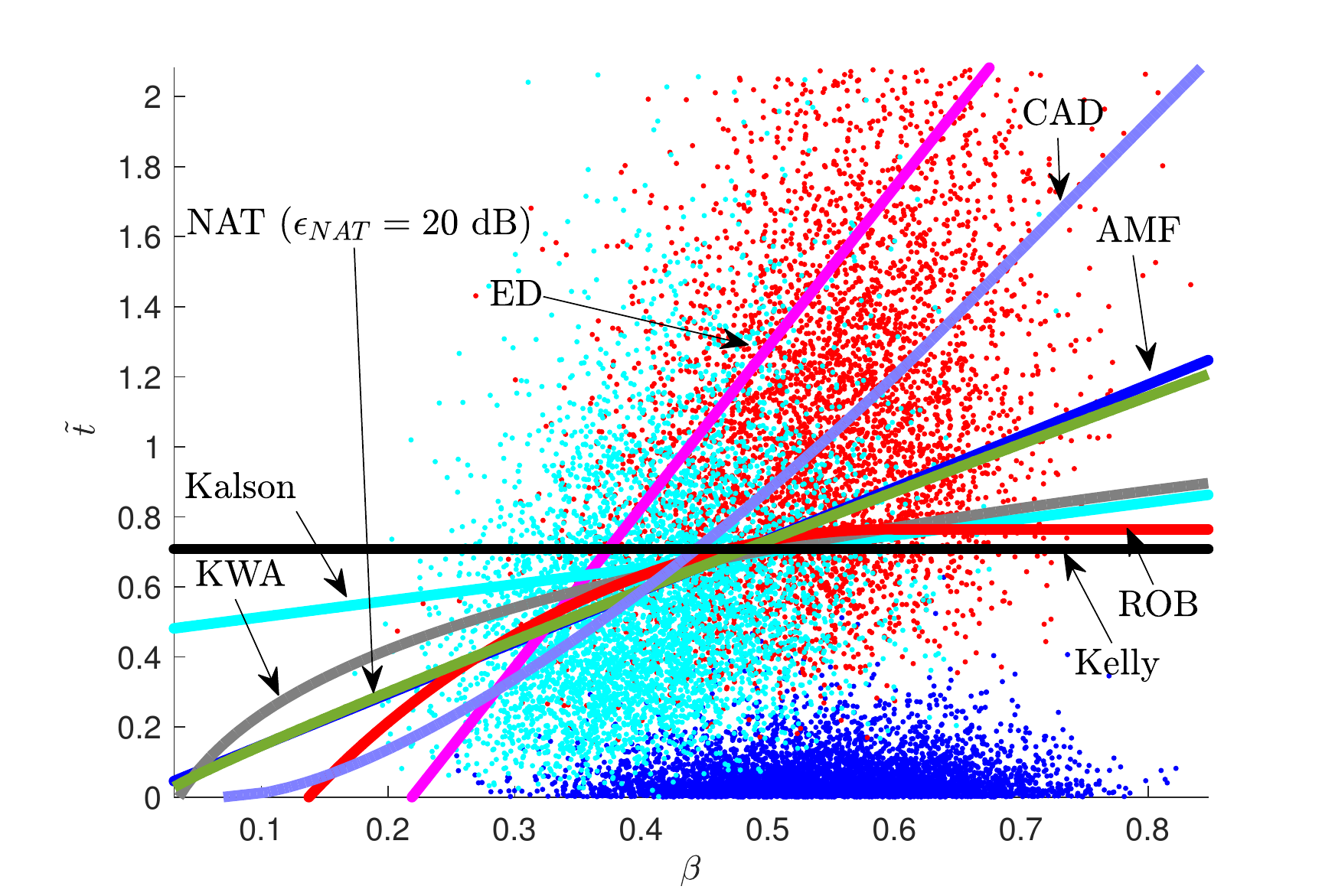}
\caption{Well-known robust receivers and proposed NAT in the CFAR-FP.}\label{fig:robustiN16K32}
\end{figure}
 \begin{figure}
\centering
\includegraphics[width=8cm]{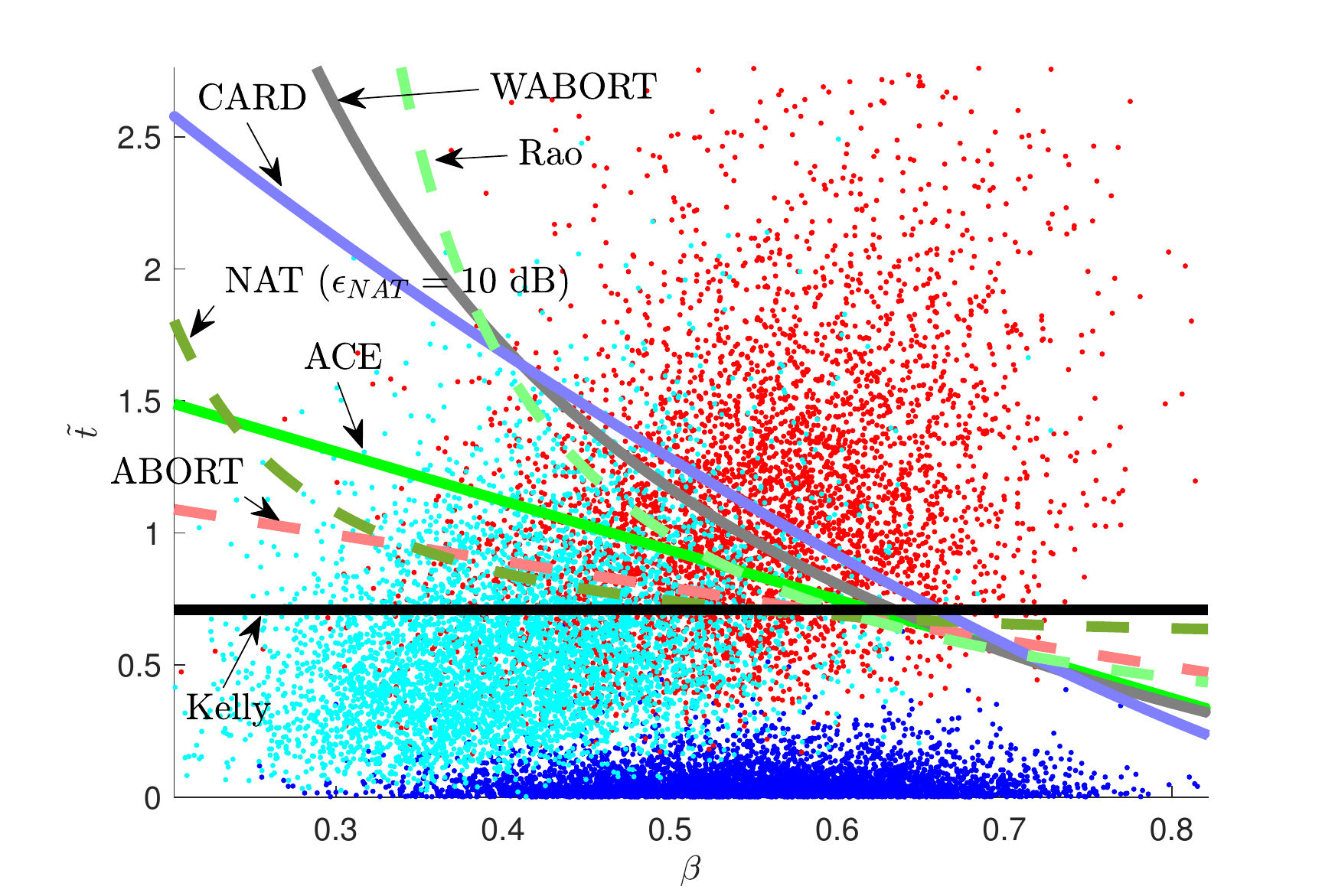}
\caption{Well-known selective receivers and proposed NAT in the CFAR-FP.}\label{fig:selettiviN16K32}
\end{figure}

The adaptive coherence estimator (ACE), also known as adaptive normalized matched filter and given by  \cite{ACE} 
\begin{align*}
t_{\text{\tiny{ACE}}}  &= 
 \frac{|\bz^{\dagger} \bS^{-1} \bv |^2}{\bv^{\dagger} \bS^{-1} \bv \  \bz^{\dagger} \bS^{-1} \bz } = \frac{\tilde{t}}{\tilde{t}+1-\beta}
\end{align*}
is  another example of selective receiver. %
Enhanced selectivity can be induced by solving a modified hypothesis testing problem that assumes the presence of a (fictitious) signal under $H_0$ either in the quasi-whitened \cite{Pulsone-Rader,Fabrizio-Farina} or whitened \cite{W-ABORT} space, so as to make it more plausible in presence of signal mismatches --- an approached called adaptive beamformer orthogonal rejection test (ABORT). Both original (quasi-whitened) ABORT and its whitened-space version (W-ABORT) can be expressed in terms of the features, respectively, as
$$
t_{\text{\tiny{A}}}  = \frac{1 +  \frac{|\bz^{\dagger} \bS^{-1} \bv |^2}{\bv^{\dagger} \bS^{-1} \bv } }{2 + \bz^{\dagger} \bS^{-1} \bz} = \frac{\tilde{t} +\beta }{\tilde{t}+1+\beta}
$$
and
\be
t_{\text{\tiny{WA}}}  = \frac{1}{(1+\bv^{\dagger} \bS^{-1} \bv) \left[ 1 -  \frac{|\bz^{\dagger} \bS^{-1} \bv |^2}{\bv^{\dagger} \bS^{-1} \bv \, (1+ \bz^{\dagger} \bS^{-1} \bz )} \right]^2}= \beta (1+\tilde{t}). \label{eq:WA}
\ee 

A further type of receivers is based on the idea of inserting a parameter in a well-known statistic, so as to obtain a tunable detector. For instance, in Kalson's detector \cite{Kalson} a nonnegative parameter, say $\epsilon_{\text{\tiny{Kalson}}}$, is introduced in the Kelly's statistic, i.e.,
\begin{align}
t_{\text{\tiny{Kalson}}}  &= 
 \frac{|\bz^{\dagger} \bS^{-1} \bv |^2}{\bv^{\dagger} \bS^{-1} \bv \, (1+ \epsilon_{\text{\tiny{Kalson}}} \bz^{\dagger} \bS^{-1} \bz )}. \label{eq:Kalson}
\end{align}
In doing so,  the degree to which mismatched signals are rejected can be controlled in between the AMF and Kelly's detector ($\epsilon_\text{\tiny Kalson}=0.5$ has been used in Fig. \ref{fig:robustiN16K32}). Kalson's detector can be rewritten in terms of the features by noticing that $t_{\text{\tiny{Kalson}}} = \frac{t_{\text{\tiny{AMF}}} }{1 + \epsilon_{\text{\tiny{Kalson}}}t_{\text{\tiny{ED}}} }$ where
\begin{equation}
t_{\text{\tiny{ED}}}  =  \bz^{\dagger} \bS^{-1} \bz \label{eq:ED}
\end{equation}
is the (adaptive) energy detector (ED). It is worth noticing that the ED totally ignores the information on the steering $\bm{v}$, hence when used alone as detection statistic will lead to reduced detection power but strong robustness.
A different tunable receiver has been proposed in \cite{KWA} by substituting the exponent of the square brackets in \eqref{eq:WA} with $2\epsilon_{\text{\tiny{KWA}}}$, which encompasses as special cases statistics equivalent to Kelly's  and  W-ABORT detectors and that, for $\epsilon_{\text{\tiny{KWA}}} < 1/2$ behaves as a robust detector but with more limited loss of $P_d$ under matched conditions, reaching the ED only as $\epsilon_{\text{\tiny{KWA}}}\rightarrow 0$. 
 
So far, a first insight can be obtained from Table \ref{tabella}:
\begin{insight}\label{ins1}
Kelly's, AMF, ACE, ED, Kalson's, and ABORT detectors are linear classifiers in the CFAR-FP: robust detectors have positive slope, while selective ones have negative slope; Kelly's detector is an horizontal line (zero slope).
\end{insight}

It can be noticed that, as quantitatively characterized in Sec. \ref{sec:characterization}, the position, shape, and orientation of the clusters change with SNR and $\cos^2 \theta$; in this respect, the positive or negative slope is suitable  to classify mismatched points as $H_1$ or $H_0$, respectively,  trading-off to some extent the achievable $P_d$ (under matched conditions) for a same SNR. Similarly, an horizontal line can effectively separate the $H_0$ cluster, since it looks horizontally spread, from any $H_1$ cluster (under matched condition) that lies in the upper part of the CFAR-FP.
We will come back on this interpretation later in Sec. \ref{sec:characterization}, after we will have provided a statistical characterization of the clusters. %

The horizontal line of Kelly's detector is obviously a consequence of choosing $\tilde{t}$ (which is equivalent to $t_\text{\tiny Kelly}$) as the ordinate axis of the CFAR-FP. More in general, other choices of the maximal invariant would yield different shapes/locations for the clusters and families of curves for the decision region boundary. 
Our choice is such that all the detectors mentioned so far, which have intimate relationship among each other, belong to the same family of linear classifiers (Insight \ref{ins1}). As an example, in terms of the maximal invariant adopted in \cite{Bose2} the AMF would have a non-linear decision region boundary, despite as known Kelly's and AMF receivers are obtained by solving the same problem under the one-step and two-step GLRT approach, and the $P_d$ of the latter tends to the $P_d$ of the former for sufficiently large sample $K$. 
Another reason to prefer the maximal invariant $(\beta, \tilde{t})$ is its widespread adoption in the literature to derive $P_{fa}$ and $P_d$ formulae, hence by using this parametrization the new results in the present paper can be more directly linked to the existing work on radar detection.

More sophisticated receivers are often non-linear classifiers in the CFAR-FP; the W-ABORT for instance (seventh row of Table \ref{tabella}) has an hyperbolic boundary $\tilde{t} \propto 1/\beta$, and similarly the KWA as well as Rao's test  \cite{Rao} given by
$$
t_{\text{\tiny{Rao}}}  = 
 \frac{|\bz^{\dagger} (\bS + \bz\bz^\dagger)^{-1} \bv |^2}{\bv^{\dagger} (\bS + \bz\bz^\dagger)^{-1} \bv } = t_{\text{\tiny{Kelly}}} \beta = \frac{\tilde{t}}{1+\tilde{t}} \beta
$$
 have hyperbolic-like boundaries; it follows that:
 \begin{insight}
 A non-linear boundary in the CFAR-FP (e.g., hyperbolic-like as in W-ABORT, KWA, and Rao) can be used to achieve a different trade-off between $P_d$ (under matched conditions) and behavior under mismatched conditions; in particular, by following more closely the cluster shapes, it may be possible to better isolate the $H_0$ cluster from the union of the matched and mismatched $H_1$ clusters, in case robustness is of interest, or   the union of the $H_0$ and mismatched-$H_1$ clusters from the matched-$H_1$ cluster, in case selectivity is desired.
 \end{insight}
Indeed, it is well-known that in several receivers enhancing robustness or selectivity often comes at the price of a certain $P_d$ loss under matched conditions compared to Kelly's receiver, with different trade-offs. 
W-ABORT for instance has strong selectivity, which however comes at the price of a reduced detection power under matched conditions. A diversified trade-off can be obtained by the KWA by tuning its parameter, which basically controls the curvature of the hyperbola-like shape ($\epsilon_\text{\tiny KWA}=0.4$ has been used in the figures).

A different idea to control the level of robustness or selectivity is to add a cone for acceptance and acceptance-rejection (hence, under the $H_1$ hypothesis only or under both hypotheses, respectively); the resulting detectors are referred to as CAD and CARD, respectively \cite{Coni-SOC,Bandiera-DeMaio-Ricci-CONI,BOR-MorganClaypool}.
More specifically, the CAD is a robust detector, encompassing as limiting cases the AMF and ED for $\epsilon_\text{\tiny CAD} \rightarrow 0$  and $\epsilon_\text{\tiny CAD} \rightarrow +\infty$, respectively; the CARD has conversely a selective behavior, whose extent can be adjusted by tuning the parameter $\epsilon_\text{\tiny CARD}$.
Both such detectors, however, experience a certain loss under matched conditions, as other receivers based on the cone idea \cite{Coni-SOC,Besson1}.
Visualized in the CFAR-FP for $\epsilon_\text{\tiny CAD}=\epsilon_\text{\tiny CARD}=0.5$, the implicit equations describing their decision boundaries appear as mildly bending lines with significant positive or negative slopes.
This explains why they cannot escape the basic trade-off of classical approaches that are linear classifiers in the CFAR-FP.
In \cite{arxivROB} it is shown, conversely, that it is possible for a CFAR detector to guarantee practically zero loss under matched conditions while providing variable robustness to mismatches, depending on the setting of a tunable parameter.
Such a parametric detector is obtained by considering under $H_1$ a fictitious signal with unknown power; in doing so, in case of matched signature no component will be  likely found and the conventional Kelly's statistic is recovered; conversely, if some mismatch  is captured by  the fictitious signal, the detector will likely find $H_1$ more plausible, making  the detector more robust.
It is thus interesting how this random-signal robustified detector (ROB) behaves in the CFAR-FP; interestingly, by equaling to a threshold its statistic
$$
t_\text{\tiny ROB} =
\left\{
\begin{array}{ll}
\frac{\left(1+\bz^{\dagger} \bS^{-1} \bz\right) \left( 1-\frac{1}{\zeta_{\epsilon}} \right)
}{\left[ \left(\zeta_{\epsilon} -1 \right)
\left\| \Pvtildeperp \tilde{\bz} \right\|^2 \right]^{\frac{1}{\zeta_{\epsilon}}} }, &
\left\| \Pvtildeperp \tilde{\bz} \right\|^2 > \frac{1}{\zeta_{\epsilon}-1} \\
\frac{1+ \left\|  \bz^{\dagger} \bS^{-1} \bz  \right\|^2}{ 1+ \left\| \Pvtildeperp \tilde{\bz} \right\|^2}, & \mbox{otherwise}
\end{array}
\right.
$$
where
$
\left\| \Pvtildeperp \tilde{\bz} \right\|^2 = \bz^{\dagger} \bS^{-1} \bz - \frac{|\bz^{\dagger} \bS^{-1} \bv|^{2}}{\bv^{\dagger} \bS^{-1} \bv} = s_2 - s_1$ and $\zeta_{\epsilon}=\frac{K+1}{N}(1+\epsilon_\text{\tiny ROB})$ with $\epsilon_\text{\tiny ROB} \geq 0$, it is possible to solve $\tilde{t}$ as an explicit function of $\beta$, reported in the second last row of Table \ref{tabella}. The resulting curve (in Fig. \ref{fig:robustiN16K32} shown for $\epsilon_\text{\tiny ROB}=0.2$) has the following interesting property:
\begin{insight}\label{ins:4}
The ROB detector \cite{arxivROB} decision region boundary in the CFAR-FP is non-linear until  $\beta=1\! -\! \frac{N}{K+1}(1+\epsilon_\text{\tiny ROB})^{-1}$, then saturates to a constant (horizontal line).
\end{insight}
This two-region behavior, with transition point $\beta \approx 0.6$ in Fig. \ref{fig:robustiN16K32} for the chosen parameters, explains the ability of ROB detector to achieve the same $P_d$ of Kelly's detector under matched conditions and at the same time strong robustness. We will come back again on this point later in Insight \ref{ins:9}.

\subsection{Derivation of natural parametric detector in the CFAR-FP}\label{sec:NAT}

Since all detectors that can be expressed as a curve in the $\beta$-$\tilde{t}$ plane are CFAR, it would be interesting to look at how it behaves a detector obtained as likelihood ratio test using such features as observations. 
This is a ``natural'' design criterion for our framework, and will ultimately lead to the parametric ``NAT'' detector reported in the last row of Table \ref{tabella}, as discussed below.
In the following we exploit the statistical characterization of the variables $\beta$ and $\tilde{t}$ that can be found in \cite{Kelly_techrep}, see also \cite{Kelly,Kelly89,BOR-MorganClaypool,KellyTR2,Richmond}.
We first notice that under $H_0$, thanks to independence, the joint distribution of $(\beta,\tilde{t})$ is given by the product of their marginal distribution, i.e., $p(\tilde{t},\beta | H_0) = p(\tilde{t} | H_0) p(\beta)$ where 
\begin{equation}
p(\tilde{t} | H_0) = \frac{(K-N+1)!}{(K-N)!} \frac{1}{(1+\tilde{t})^{K-N+2}}
\end{equation}
is a complex central F-distribution with $1$ and $K-N+1$ (complex) degrees of freedom, denoted by $\mathcal{CF}_{1,K-N+1}$, and 
\begin{equation}
p(\beta) = \frac{K!}{(N-2)! (K-N+1)!} \beta^{K-N+1} (1-\beta)^{N-2}
\end{equation}
is a complex central Beta distribution with $K-N+2$ and $N-1$ degrees of freedom, denoted by $\mathcal{C\beta}_{K-N+2,N-1}$.
Under $H_1$, we have the product of the same $p(\beta)$ and a complex noncentral F-distribution with $1$ and $K-N+1$ complex degrees of freedom and noncentrality parameter $\gamma \beta$, in symbols $\mathcal{CF}_{1,K-N+1}(\gamma \beta)$:
\begin{align}
p(\tilde{t},\beta | H_1) &= \frac{(K-N+1)!}{(K-N)!} \frac{1}{(1+\tilde{t})^{K-N+2}}  \mathrm{e}^{-\frac{\gamma \beta}{1+\tilde{t}}} \nonumber\\
&\times   \sum_{h=0}^{K-N+1}  \frac{\binom{K-N+1}{h}}{h!} \left( \frac{\gamma\beta \tilde{t}}{1+\tilde{t}} \right)^h  p(\beta).
\end{align}
As a consequence, the likelihood ratio test statistic is
\begin{equation}
t_\text{\tiny MPI} = \frac{p(\tilde{t},\beta | H_1)}{p(\tilde{t},\beta | H_0)} = \mathrm{e}^{-\frac{\gamma \beta}{1+\tilde{t}}}  \sum_{h=0}^{K-N+1} a_h(\gamma) \left( \frac{\beta \tilde{t}}{1+\tilde{t}} \right)^h \label{eq:MPI}
\end{equation}
where $a_h(\gamma)=\frac{\binom{K-N+1}{h}}{h!} \gamma^h$. 
We denote such a detector by $t_\text{\tiny MPI}$ since it is easy to show that it is equivalent to the MPI test in \cite{Bose2}, which maximizes $P_d$ for fixed SNR $\gamma$ (it is the Neyman-Pearson test for the observables $(\tilde{t},\beta)$ with known SNR) while both the uniformly most powerful (UMP) test and the UMP test in this invariant class (so-called UMPI) do not exist.
Clearly, \eqref{eq:MPI} cannot be implemented since $\gamma$ is unknown in practice, hence it represents only a theoretical bound for the achievable performance in term of $P_d$ (under matched conditions).
However, it is noticed  in \cite{Bose2} that the gap with respect to Kelly's GLRT is very small, which supports the widespread adoption of the latter as performance benchmark for matched conditions, despite no optimality is guaranteed by the GLRT procedure. Under the CFAR-FP framework, this has a geometrical interpretation (see discussion on Insight \ref{ins1}, further developed in Sec. \ref{sec:characterization}).

To come up with an implementable test, the solution in \cite{Bose2} is to derive the test that maximizes the first-order Taylor approximation of $P_d$  in the neighborhood of a chosen SNR, so obtaining a \emph{locally} most powerful invariant (LMPI) test; this is applied in particular to the detection of weak signals ($\gamma \approx 0$). 
Here we take a different path, aiming at obtaining a detector with more general applicability and, also, controllable behavior under mismatched conditions (as already mentioned, mismatches are not considered in \cite{Bose2}). %

The ad-hoc solution we propose is to replace $\gamma$ in \eqref{eq:MPI} with a tunable parameter $\epsilon_\text{\tiny NAT}$. This simple idea yields a new parametric detector $t_\text{\tiny NAT}$ with interesting properties, and was suggested by the analysis of the decision region boundary of $t_\text{\tiny MPI}$ in the CFAR-FP: Fig. \ref{fig:robustiN16K32} reveals for instance that for $\epsilon_\text{\tiny NAT}=100$ (20 dB) it has a practically linear shape and resembles a robust detector; conversely, for $\epsilon_\text{\tiny NAT}=10$ (10 dB) it has a convex shape with high curvature and appears as a very selective detector (see Fig. \ref{fig:selettiviN16K32}). 
A more detailed analysis, reported in Fig. \ref{fig:NAT_regions}, shows that between these two curves the shape of the decision region boundary changes with continuity; that is, interestingly, the selective detector progressively metamorphoses into a robust one, as demonstrated in Fig. \ref{fig:NAT_snr} in terms of $P_d$ vs SNR under matched and mismatched conditions.\footnote{A different idea could be to consider the GLRT  directly to the CFAR-FP, by maximizing  \eqref{eq:MPI} with respect to $\gamma$. This can be done by a numerical search or by computing the roots of the polynomial of degree $N-K+2$ that is obtained by imposing the derivative  of \eqref{eq:MPI} equal to zero. In all our tests, the resulting $P_d$ is the same of Kelly's detector under both matched and mismatched conditions, so we opted for the NAT idea.}

The performance of NAT for the two extrema of the interval (i.e., 10 dB and 20 dB, taken as representative of two opposite behaviors) will be further analyzed  in Sec. \ref{sec:perf_ass} in comparison with all the other detectors. From Fig. \ref{fig:NAT_snr} it is though already evident that the NAT detector under matched conditions is very close to Kelly's detector, unless $\epsilon_\text{\tiny NAT}$ is set too large; remarkably, this allows one to obtain detectors as powerful as Kelly's GLRT but more selective or more robust compared to it. Actually, for matched conditions the NAT even outperforms Kelly's detector in correspondence of the specific SNR equal to $\epsilon_\text{\tiny NAT}$: this is obvious since, in that case, the NAT receiver coincides with the MPI \eqref{eq:MPI}. However, as already noticed, the gap is very small: in Fig. \ref{fig:NAT_snr} the loss is only 0.1 dB for the same $P_d$. 
As concerns the performance under mismatched conditions, the NAT receiver has the great advantage to flexibly steer the behavior towards  robustness or selectivity, while keeping under matched conditions practically the same $P_d$ of Kelly's detector; the MPI, if it were implementable, would be more robust than Kelly's detector especially at higher SNR values (reported as a dashed curve in Fig. \ref{fig:NAT_snr}).

In addition to the NAT detector discussed above, we will provide in Sec. \ref{sec:perf_ass} different criteria to design directly in the CFAR-FP detectors with desired robust or selective behavior.

\begin{figure}
\centering
\includegraphics[width=7cm]{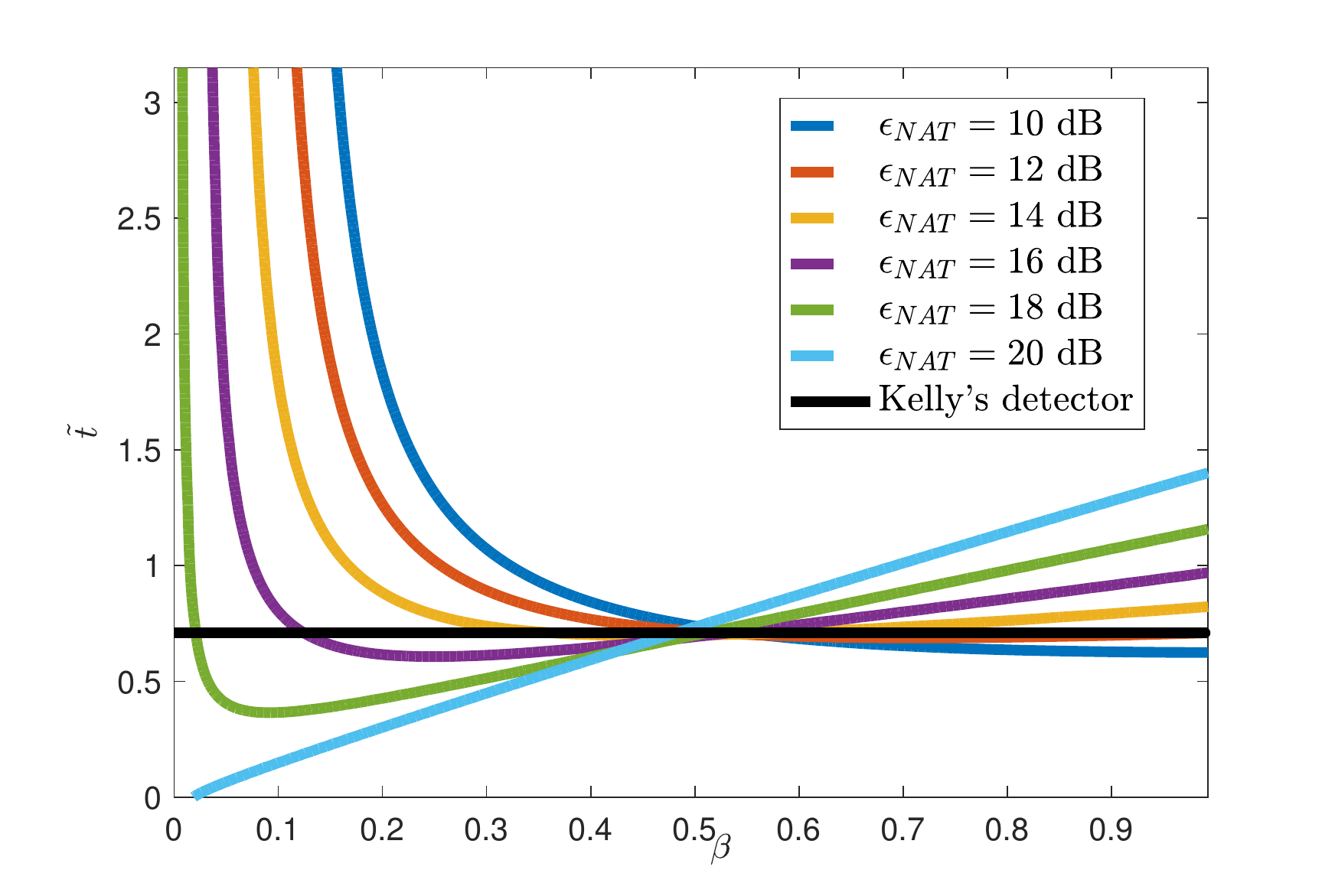}
\caption{Decision region boundaries for the NAT receiver.}\label{fig:NAT_regions}
\end{figure}
\begin{figure}
\centering
\includegraphics[width=7cm]{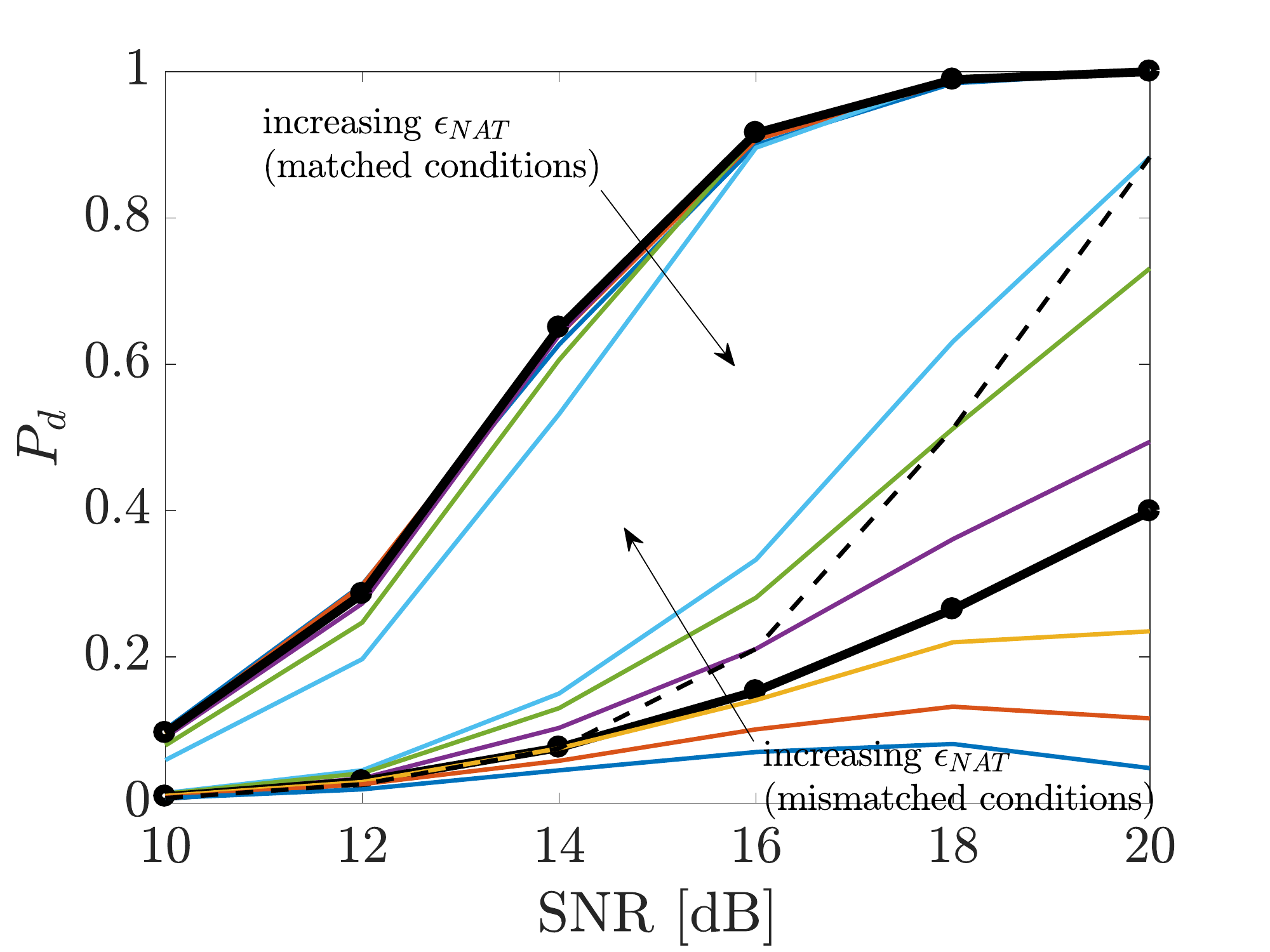}
\caption{$P_d$ vs SNR for different values of  $\epsilon_\text{\tiny NAT}$ from 10 to 20 dB in steps of 2 dB (as in Fig. \ref{fig:NAT_regions}), under both matched and mismatched conditions. Kelly's detector is shown as a thicker black line with marker. The dashed line is the MPI under mismatches (its $P_d$ under matched conditions, not visible, is practically the same of Kelly's detector).}\label{fig:NAT_snr}
\end{figure}

\section{Statistical characterization of data clusters in the CFAR-FP}\label{sec:characterization}

We provide closed-form expressions for the mean and covariance matrix of $(\beta, \tilde{t})$, which are useful to understand the position of the data clusters in the CFAR-FP and assess how they change as function of $\gamma$ and $\cos^2 \theta$. This will ultimately lead to the definition of further design criteria in the CFAR-FP.

\begin{proposition}\label{teo:cluster}
The cluster of the points $\bm{\xi} = [\beta \ \tilde{t}]^T$, built from $\bz \in \C^{N \times 1}$ and $K$ secondary data $\bZ = [\bz_1 \cdots \bz_{K}] \in \C^{N \times K}$ through the mapping \eqref{eq:FS_mapping}, has in general a slanted elliptical shape in the CFAR-FP, with  parametric equation 
$$ 
[\mu_\beta \ \mu_{\tilde{t}}]^T +  \bm{U} \bm{\Lambda}^{1/2} [\cos s \ \sin s]^T, \quad  s\in[0,1]
$$ 
where the center $\E[\bm{\xi}] = [\mu_\beta \ \mu_{\tilde{t}}]^T$ is given by
\begin{align}
\mu_\beta &=  1 - \frac{N-1}{K+1} \e^{-\gamma (1-\cos^2 \theta)} \nonumber\\
& \times  {}_2F_2\left( K+1, N; N-1, K+2; \gamma (1-\cos^2 \theta) \right) \label{eq:mu_beta}
\end{align}
\begin{equation}
\mu_{\tilde{t}} = \frac{1+ \gamma \mu_\beta  \cos^2 \theta }{K-N}  \label{eq:mu_t}
\end{equation}
and axes/orientation  are obtained from the covariance matrix
\begin{equation}
\mathrm{COV}[\bm{\xi}] =  \left[ \begin{array}{cc} \sigma^2_\beta & \rho \sigma_\beta \sigma_{\tilde{t}} \\  \rho \sigma_\beta \sigma_{\tilde{t}} & \sigma^2_{\tilde{t}}  \end{array} \right] =\bm{U}\bm{\Lambda}\bm{U}^T
\end{equation}
with
\begin{align}
\sigma^2_\beta &=  \frac{N(N-1)}{(K+2)(K+1)}  \e^{-\gamma (1-\cos^2 \theta)}  \nonumber\\
 \times{}_2F_2 & \!\left( K+1, N+1; N\!- \! 1, K+3; \gamma (1\!- \!\cos^2 \theta) \right) \! -\!  (1\!- \!\mu_\beta)^2 \label{eq:var_beta}
\end{align}
\begin{align}
\sigma^2_{\tilde{t}} &=  \frac{(\gamma \cos^2 \theta)^2 (\sigma^2_\beta +\mu^2_\beta) + (1+2 \gamma \mu_\beta \cos^2 \theta ) (K-N+1)}{(K-N)^2(K-N-1)}  \nonumber\\
& +  \frac{(\gamma \cos^2 \theta)^2 \sigma^2_\beta}{(K-N)^2}  \label{eq:var_t_tilde}
\end{align}
\begin{equation}
\rho =  \frac{1}{\sigma_\beta \sigma_{\tilde{t}}} \left[ \frac{\mu_\beta}{K-N} + \frac{\gamma   (\sigma^2_\beta +\mu^2_\beta)\cos^2 \theta}{K-N}  - \mu_\beta\mu_{\tilde{t}} \right]. \label{eq:rho}
\end{equation}
$\!{}_2F_2 (a,\! b;\! c,\! d;\! x)\!$ is the generalized hypergeometric function \cite{abramowitz}.
 \end{proposition}
 \begin{proof}
 See Appendix \ref{app:A}.
 \end{proof}

Some interesting properties can be deduced from Proposition \ref{teo:cluster}, as discussed below with the aid of Figs. \ref{fig:ellipseN16K32}-\ref{fig:ellipseN16K20}. First of all, the perfect agreement between the point clouds and the theoretical $1\sigma$-ellipses drawn in dashed line can be verified. Then, a more formal understanding of the characteristics of the $H_0$ cluster empirically observed in Sec. \ref{sec:interpretation} can be obtained: 
\begin{insight}\label{ins:5}
The cluster under $H_0$ has a major axis parallel to the abscissas, i.e., it is oriented horizontally in the CFAR-FP plane, and its minor axis is quite compressed. 
\end{insight}
This can be easily seen by particularizing the expressions of Proposition \ref{teo:cluster} for $\gamma=0$, obtaining 
\begin{equation}
(\mu_\beta,\mu_{\tilde{t}}) = \left( \frac{K-N+2}{K+1}, \frac{1}{K-N} \right) \label{eq:centerH0}
\end{equation}
which in turn leads to $\rho=0$ (hence no rotation); moreover, the axes of the $1\sigma$-ellipses for $H_0$ are given by
\begin{equation}
\sigma_\beta = \frac{\sqrt{N(N-1)(K+1)-(N-1)^2(K+2)}}{(K+1)\sqrt{K+2}}\label{eq:stdH0}
\end{equation}
\begin{equation}
\sigma_{\tilde{t}} = \frac{\sqrt{K-N+1}}{(K-N) \sqrt{K-N-1}}.
\end{equation}
For increasingly large $K$ and fixed $N$, such expressions show that the $H_0$ cluster is generally quite concentrated, especially towards the vertical dimension; moreover, it turns out that $\sigma_\beta \approx \sqrt{N-1}/K$ and $\sigma_{\tilde{t}} \approx 1/(K-N)$, which means that the ratio  $\sigma_\beta/\sigma_{\tilde{t}} \approx (1-\frac{N}{K})\sqrt{N-1}$.
Thus, for concrete values of $K$ and $N$ the resulting shape is a stretched ellipsis parallel to the abscissa axis; for instance, $K=2N$ with non-small $N$ yields  $\sigma_{\tilde{t}} \approx 1/N \ll \sigma_\beta \approx 0.5/\sqrt{N}$. Clearly, we have that:
\begin{insight}\label{inscinque}
By increasing the number of secondary data $K$, the $H_0$ cluster shrinks and its center migrates towards the bottom-right corner of the CFAR-FP, becoming progressively very localized around the point $(\beta,\tilde{t})=(1,0)$ as $K\rightarrow\infty$. 
\end{insight}
This obviously reflects the increasing amount of available information that makes the $H_0$ cluster better separable from the $H_1$ one (for the comparison, notice that Figs. \ref{fig:ellipseN16K32}-\ref{fig:ellipseN16K20} are in different scales), see also the discussion about Insight \ref{ins1}.

\begin{figure}
\centering
\includegraphics[width=7cm]{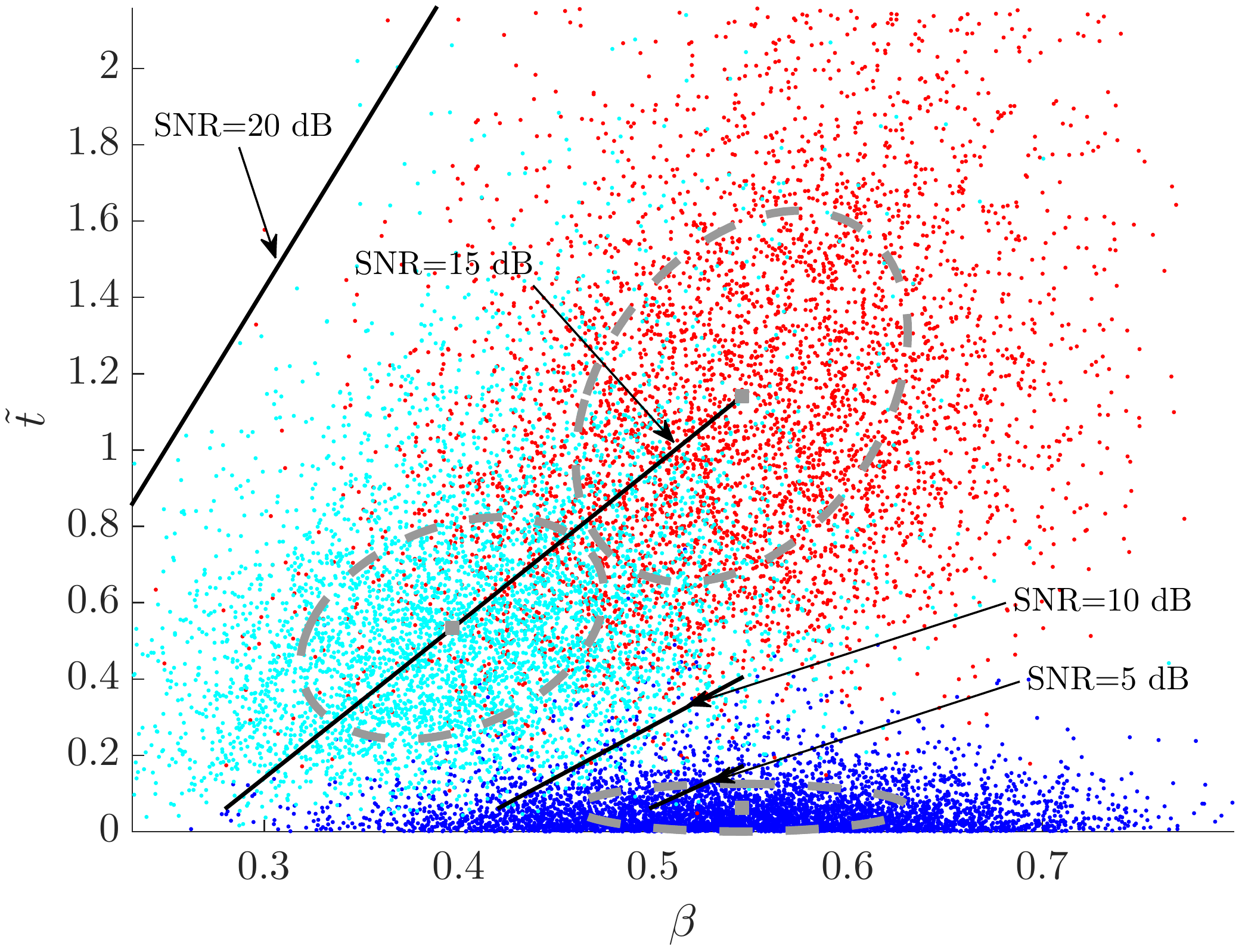}
\caption{Cluster ellipses for several values of SNR and corresponding trajectories for $\cos^2 \theta$ ranging in $[0,1]$, $K=32$.}\label{fig:ellipseN16K32}
\end{figure}
 \begin{figure}
\centering
\includegraphics[width=7cm]{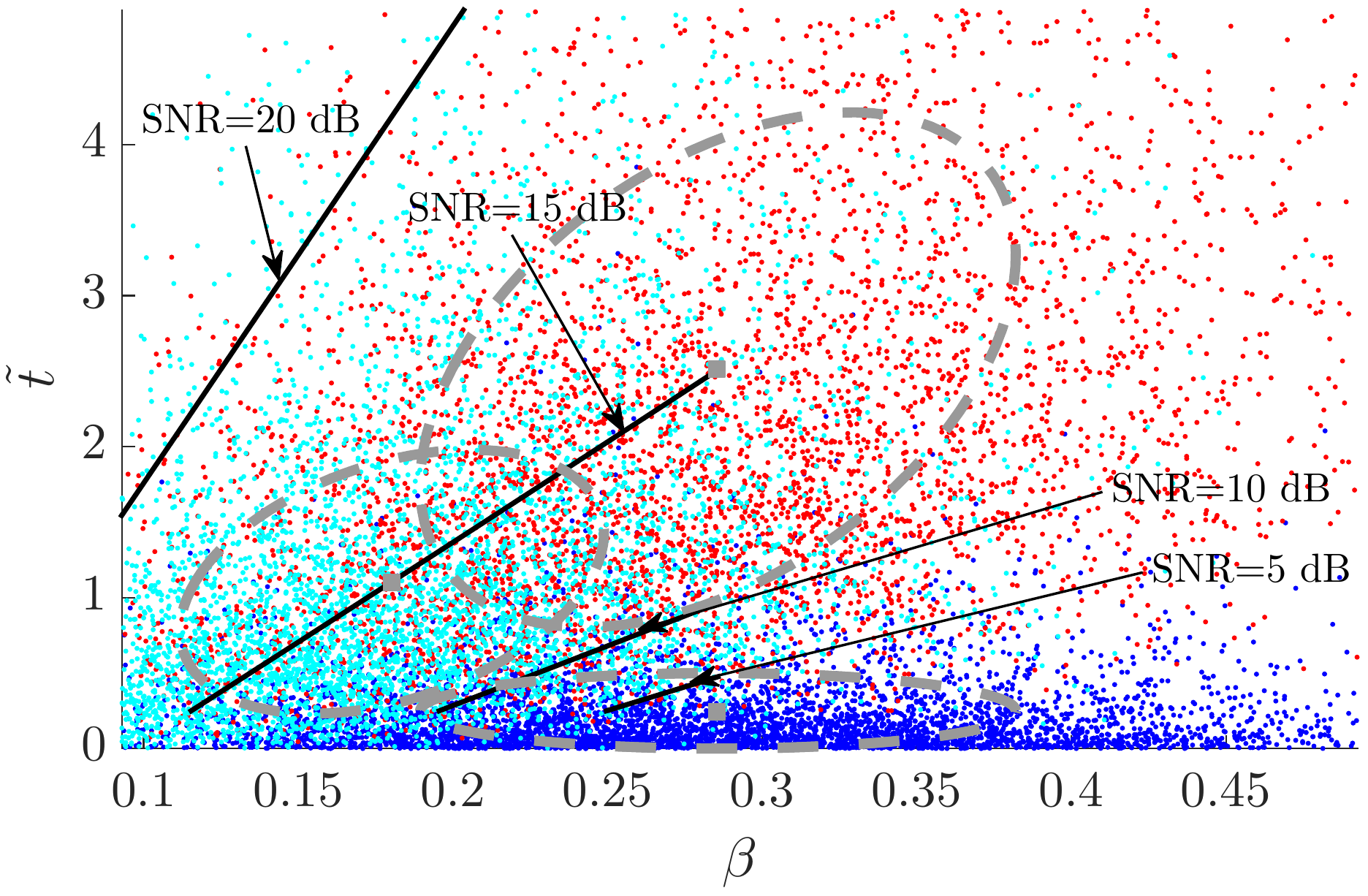}
\caption{Cluster ellipses for several values of SNR and corresponding trajectories for $\cos^2 \theta$ ranging in $[0,1]$, $K=20$.}\label{fig:ellipseN16K20}
\end{figure}

 \begin{figure}
\centering
\includegraphics[width=9cm]{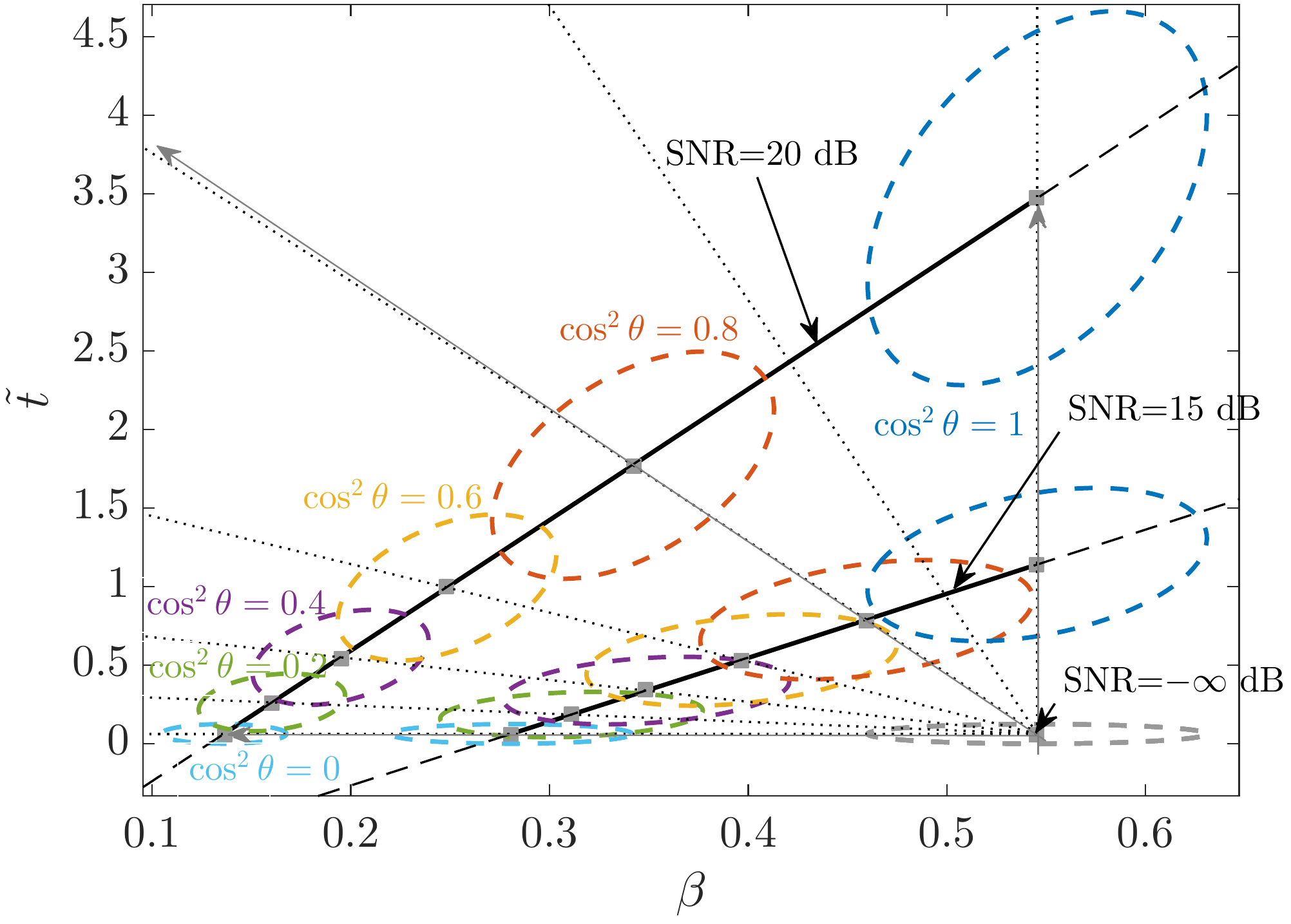}
\caption{Cluster ellipses migration for two values of SNR $\gamma$ and corresponding linear approximation of the \emph{iso-SNR} trajectories (dashed lines) for $\cos^2 \theta$ ranging in $[0,1]$, $K=32$. Dash-dot lines describe instead the \emph{iso-mismatch} trajectories of the $H_0$ cluster center as $\gamma$ increases (thus turning into $H_1$): vertical line for perfect match $\cos^2\theta=1$, horizontal line for full mismatch $\cos^2\theta=0$ (orthogonal), diagonal lines for intermediate values.}\label{fig:ellipses_migrN16K32}
\end{figure}

Another intuitive observation of Sec. \ref{sec:interpretation} that is formally proven by Proposition \ref{teo:cluster} is the following: 
\begin{insight}
The position and shape of the data point clusters in the CFAR-FP change as a function of $\gamma$ and $\cos^2\theta$. Specifically, iso-SNR (fixed $\gamma$, as function of $\cos^2\theta$) and iso-mismatch (fixed $\cos^2\theta$, as function of $\gamma$) trajectories can be drawn, which describe how the center of the cluster migrate under different conditions, while the axes shrink or expand, and rotate counterclockwise.\footnote{It is worth noticing again here, as already discussed in Sec. \ref{sec:introduction},  that by contrast in  mesa plots  iso-$P_d$ curves are drawn in a $\gamma$-$\cos^2 \theta$ plane.}
\end{insight}
In particular, by specializing the expressions in Proposition \ref{teo:cluster} for $\cos^2\theta=1$, the behavior of the $H_1$ cluster (under matched conditions) as function of $\gamma$ is readily obtained:
\begin{equation}
(\mu_\beta,\mu_{\tilde{t}}) = \left( \frac{K-N+2}{K+1}, \frac{1+\gamma\mu_\beta}{K-N} \right) \label{eq:centerH1}
\end{equation}
which of course for $\gamma=0$ returns the center of the $H_0$ cluster given in eq. \eqref{eq:centerH0}. Thus, eq. \eqref{eq:centerH1} gives the expression of the vertical line in the $\beta$-$\tilde{t}$ plane described by the center of the $H_1$ cluster as $\gamma$ increases; moreover, it is a simple matter to see that $\sigma_\beta$ is independent of $\gamma$ and identical to the value obtained for the $H_0$ cluster, while $\rho\neq 0$ and $\sigma_{\tilde{t}}$ both depend on $\gamma$, implying that, as shown in Fig. \ref{fig:ellipses_migrN16K32}:
\begin{insight}\label{ins:8}
For increasing $\gamma$, the $H_1$ cluster migrates vertically while expanding its minor axis and rotating counterclockwise.
\end{insight}

At the other extreme we have the trajectory described by the $H_1$ cluster under totally mismatched conditions; it is obtained by specializing the expressions in Proposition \ref{teo:cluster} for $\cos^2\theta=0$:
\begin{equation}
(\mu_\beta,\mu_{\tilde{t}}) = \left( 1-g(\gamma), \frac{1}{K-N} \right) \label{eq:centerH1mism}
\end{equation}
where
\begin{equation}
g(\gamma)=\frac{N-1}{K+1} \e^{-\gamma} {}_2F_2(K+1, N; N-1, K+2; \gamma) \label{eq:g}.
\end{equation}
In this case eq. \eqref{eq:centerH1mism} describes an horizontal line; moreover, both $\mu_{\tilde{t}}$ and $\sigma_{\tilde{t}}$ are independent of $\gamma$, while  $\rho=0$. Thus:
\begin{insight}\label{ins:9}
For actual steering vector orthogonal to the nominal one ($\bv^\dag \bC^{-1} \bp =0$, i.e., $\cos^2 \theta=0$), the $H_1$ cluster flattens and becomes parallel to the abscissa axis (alike $H_0$). %
\end{insight}

The two linear (vertical and horizontal) trajectories given by eqs. \eqref{eq:centerH0} and \eqref{eq:centerH1mism} can be though as ``canonical directions'' for the detection problem. In particular, as already observed in Insight \ref{ins1}, Kelly's detector has the singular property of being an horizontal line in the CFAR-FP; thus, a possible explanation of the very good detection power of Kelly's detector is that, given the characteristics of the $H_0$ cluster (see Insights \ref{ins:5} and \ref{inscinque}), an horizontal decision region boundary tends to keep the $H_0$ cluster well-separated from any $H_1$ cluster under matched conditions. On the other hand, Kelly's detector is designed without taking into account the possibility of a mismatch on the steering vector, i.e., without trying to adapt its behavior to cope with this eventuality, as conversely done by most of other detectors analyzed so far --- which in fact exhibit an oblique-linear or even non-linear boundary. An obvious question is then how a detector that is orthogonal to such a canonical horizontal (Kelly's) direction can perform, i.e., a detector based on $\beta$ alone. From the analysis of the clusters' positions, it is clear that such a solution cannot provide good detection performance, since a vertical line is unsuitable to separate the $H_0$ and $H_1$ clusters. However, we also notice that a large enough value of $\beta$ indicates with high probability that the data under test does not belong to any $H_1$ cluster under mismatched condition, since the latter tends to be confined in the left-most half of the CFAR-FP (thus, the hypothesis in force is likely either $H_0$ or $H_1$ under matched conditions). It turns out that this property is  exploited by the ROB detector \cite{arxivROB} (see Insight \ref{ins:4}), which has an increasing (non-linear) trend in the lower range of $\beta$ while it is horizontal for larger values, to achieve zero $P_d$ loss compared to Kelly's detector under matched conditions and at the same time  to be very robust to mismatches, while fulfilling the $P_{fa}$ constraint. Summarizing:
\begin{insight}
The horizontal ($\tilde{t}=\text{const}$) and vertical ($\beta=\text{const}$) decision region boundaries are ``canonical'' directions in the CFAR-FP: the former basically acts as discriminative feature for the $H_1$ cluster under matched conditions; the latter basically acts as discriminative feature for the $H_1$ cluster under mismatched conditions; both of them are needed to promote either robustness or selectivity in the detector.
\end{insight}

A final and remarkable insight is as follows. Notice  that  the iso-mismatch trajectories of $(\mu_\beta(\gamma), \mu_{\tilde{t}}(\gamma))$ obtained for chosen values of $\cos^2\theta \neq 0,1$  as function of $\gamma$ are generally non-linear: their equations are indeed those of Proposition \ref{teo:cluster}, reported as dotted lines  in Fig. \ref{fig:ellipses_migrN16K32}. However, it is apparent that they are very close to straight lines (an example is displayed with an arrow for $\cos^2=0.8$); similarly, as the mismatch parameter $\cos^2 \theta$ varies in $[0,1]$ for fixed $\gamma$, the center of the cluster describes in the CFAR-FP an iso-SNR trajectory that is linear with excellent approximation (compare dashed lines with solid lines in Fig. \ref{fig:ellipses_migrN16K32}), despite the highly non-linear relationship given by the generalized hypergeometric function in eq. \eqref{eq:mu_beta}. This non-trivial fact is stated and proved below.

\begin{proposition}\label{teo:linear_migration}
The iso-SNR curves described  by $(\mu_\beta,\mu_{\tilde{t}})$ in the CFAR-FP for fixed $\gamma$ as function of $\cos^2 \theta\in [0,1]$, as defined in Proposition \ref{teo:cluster}, are well approximated by the  line
\begin{equation}
\tilde{t} = m \beta + q
\end{equation}
where
\begin{equation}
m=\frac{K+2-N}{(K-N)[(K+1)g(\gamma) +1 - N]} \gamma \label{eq:m}
\end{equation}
$$q=m \left[ g(\gamma)-1 \right] + \frac{1}{K-N}$$
and $g(\gamma)$ is given by eq. \eqref{eq:g}.
A similar result can be obtained for iso-mismatch trajectories (fixed $\cos^2\theta$ as a function of $\gamma$).
\end{proposition}
 \begin{proof}
See Appendix \ref{app:B}.
 \end{proof}

\section{Design and performance assessment of radar detectors in the CFAR-FP}\label{sec:perf_ass}

From the comprehensive discussion of the CFAR-FP framework performed above, it should be clear that the $\beta$-$\tilde{t}$ plane can be exploited also as a design tool. A first detector we have proposed is the NAT detector given in Sec. \ref{sec:NAT}, which is derived based on the likelihood ratio test computed directly for the ``observables'' $(\beta,\tilde{t})$. %
In general, ad-hoc detectors can be designed by empirically drawing the desired decision region boundary in the CFAR-FP; then, the equation  that will ultimately become the decision statistic is obtained through curve fitting or other means. 
Though appealing, such an approach has two main drawbacks. First, it lacks a practical mechanism to control the $P_{fa}$: %
indeed, given an arbitrary decision region boundary in the CFAR-FP, the resulting $P_{fa}$ can be computed from the characterization of the $H_0$ cluster, either analytically or numerically; however, in order to come up with a detector working at a preassigned $P_{fa}$, one has to iteratively modify by trial-and-error the decision region boundary until the $P_{fa}$ constraint is fulfilled. Second, this time-consuming procedure would anyway yield a very tailored ad-hoc detector, meaning that the whole procedure must be repeated from scratch if any of the parameters is changed.

A more methodological approach can be identified by observing that any detector admitting an explicit function for its decision region boundary (e.g., all but three in Table \ref{tabella}) can be equivalently expressed in terms of the statistic
\begin{equation}
t = \tilde{t} - f(\beta, \bm{\epsilon})\label{eq:struct}
\end{equation}
with $\bm{\epsilon}$ a possible vector of parameters the detector depends upon and $f(\cdot)$ a known function of $\beta$ parameterized in $\bm{\epsilon}$. 
\begin{figure}
\centering
\includegraphics[width=8cm]{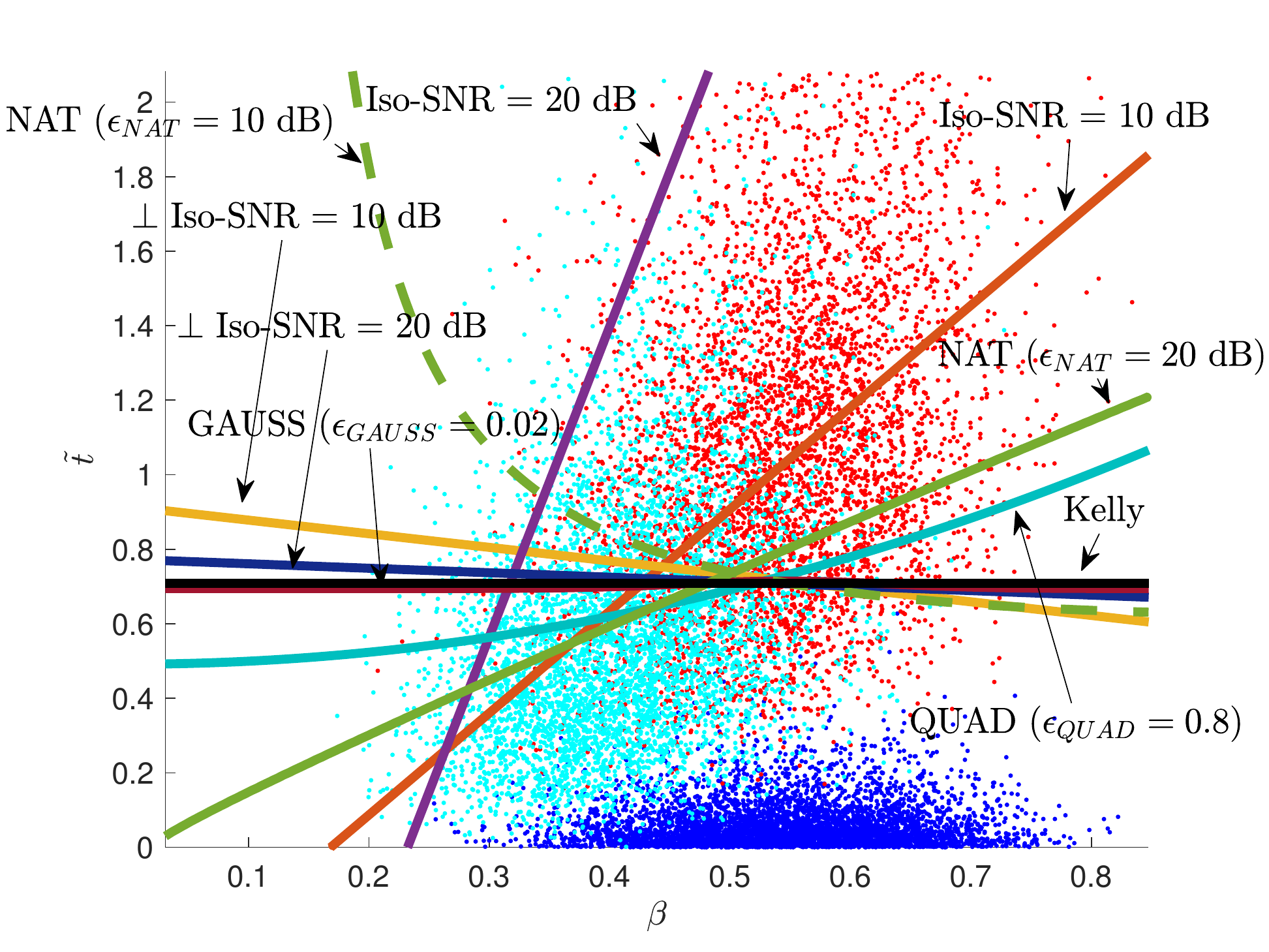}
\caption{Decision region boundaries of the different proposed detectors in the CFAR-FP, in comparison with Kelly's detector (notice that Iso-SNR and $\perp$Iso-SNR do not look orthogonal only due to the different scales of the axes).}
\label{fig:regioni_proposed}
\end{figure}
To overcome the drawbacks of the previously discussed trial-and-error procedure, we propose a more practical strategy: choose a model for $f$ depending on a small number of parameters $\bm{\epsilon}$  that are set upfront, leaving $\eta$ as the only one to be optimized to satisfy the $P_{fa}=\psi$ constraint. The threshold $\eta$ is in fact the constant term of the decision boundary function, since any other constant term added to $f$ can be absorbed into $\eta$. So, for a detector $X$, the  final test can be always written as
$$
\tilde{t} - f(\beta,\bm{\epsilon}_X)  \test \eta_X.
$$
Simple choices for $f$ are a polynomial $f=\sum_{i=1}^p \epsilon_i \beta^i$ with low order $p$ or the Gaussian function $f=\epsilon_1 e^{-\frac{(\beta-\epsilon_2)^2}{\epsilon_3}}$ with tunable height $\epsilon_1$, location $\epsilon_2$, and width $\epsilon_3$ parameters. Still, setting two or three parameters without any guideline is a non-trivial task, hence special cases of such functions with only one degree of freedom are preferable, i.e., $\bm{\epsilon}_X$ is a scalar $\epsilon_X$. Examples are:
\begin{itemize}
\item the quadratic function without the linear term and with upward concavity, leading to the decision boundary 
$$
\tilde{t} -  \epsilon_\text{\tiny QUAD} \beta^2 = \eta_\text{\tiny QUAD}, \quad \epsilon_\text{\tiny QUAD}>0
$$ 
which is a parabolic (convex) arc in the CFAR-FP thus separating the $H_0$ cluster from the $H_1$ cluster under both matched and mismatched conditions (expectedly producing a robust behavior), see Fig. \ref{fig:regioni_proposed};
\item the Gaussian function with location parameter $\epsilon_2=\mu_\beta$ from eq. \eqref{eq:centerH0} and width parameter $\epsilon_3=2\sigma_\beta^2$ with $\sigma_\beta$ from eq. \eqref{eq:stdH0}, leading to the decision region boundary 
$$
\tilde{t} -  \epsilon_\text{\tiny GAUSS} \e^{-\frac{(\beta-\mu_\beta)^2}{2\sigma^2_\beta}} = \eta_\text{\tiny GAUSS}, \quad \epsilon_\text{\tiny GAUSS}>0
$$ 
aimed at better enclosing the $H_0$ cluster; this is motivated by the fact that the marginal distribution $p(\beta)$ (under $H_0$) is similar to a normal density $\mathcal{N}(\mu_\beta,\sigma^2_\beta)$ (figure omitted due to lack of space)\footnote{Clearly, the location parameter is chosen as the center of the $H_0$ cluster; as to the width parameter, a reasonable value should capture almost the whole support of the $H_0$ cluster, hence can be linked to the standard deviation of $\beta$. Other choices are of course possible, anyway equally or even more heuristic.}, hence the Gaussian shape can follow the dispersion of the points as much as the vertical shift to satisfy the $P_{fa}$ constraint allows to (see Fig. \ref{fig:regioni_proposed});
\item the linear function $y=\epsilon_\text{\tiny LIN} \beta$, so that the decision region boundary in the CFAR-FP is 
$$
\tilde{t} -  \epsilon_\text{\tiny LIN} \beta = \eta_\text{\tiny LIN}, \quad \epsilon_\text{\tiny LIN}\in \R.
$$ 
\end{itemize}

 We remark that in these examples only $\epsilon_X$ must be manually set, while $\eta_X$ is obtained by inverting the $P_{fa}$ formula. However, while the setting of $\epsilon_X$ for the quadratic and Gaussian functions remains ad-hoc, for the linear function we are able to provide a criterion that can serve as a guideline for setting the slope $\eta_\text{\tiny LIN}$.
The idea is to use reference directions that have a specific meaning in the CFAR-FP and, also, are parametric in $N,K$ so that the resulting detection statistic is general. Proposition \ref{teo:linear_migration} can be exploited to this aim: in particular, iso-SNR lines are provided in closed-form and have an intuitive interpretation that can help in the choice. Specifically, knowing how the $H_1$ cluster migrates for increasing mismatch allows one to try to keep it ``above threshold'' so obtaining a robust detector; conversely, setting the decision region boundary orthogonally to a iso-SNR line, will produce a $H_0$ decision (rejection) as soon as the mismatch exceeds a certain level, thus yielding a selective behavior. So, $\epsilon_\text{\tiny LIN}$ can be set equal to $m$ in eq. \eqref{eq:m} (for chosen $\cos^2 \theta$) to obtain a robust detector; instead, $\epsilon_\text{\tiny LIN} = -1/m$ will lead to a selective detector (labeled ``$\perp$iso-SNR''). The performance of such receivers will be more quantitatively assessed below in terms of $P_d$ vs SNR under matched and mismatched conditions.

Beforehand, as a final note, we observe that other reference directions can be considered, in particular iso-mismatch lines. However, their expressions are slightly more involved; %
moreover, it is easy to realize that considering the $\perp$iso-SNR together with the iso-SNR lines spans all possible directions of the plane, hence iso-mismatch curves would not add more possibilities (just a different interpretation). For these reasons, we stick to iso-SNR lines, hence omitted the explicit expressions of iso-mismatch lines in Proposition \ref{teo:linear_migration}.

\begin{figure}
\centering
\subfigure[matched]{\includegraphics[width=8cm]{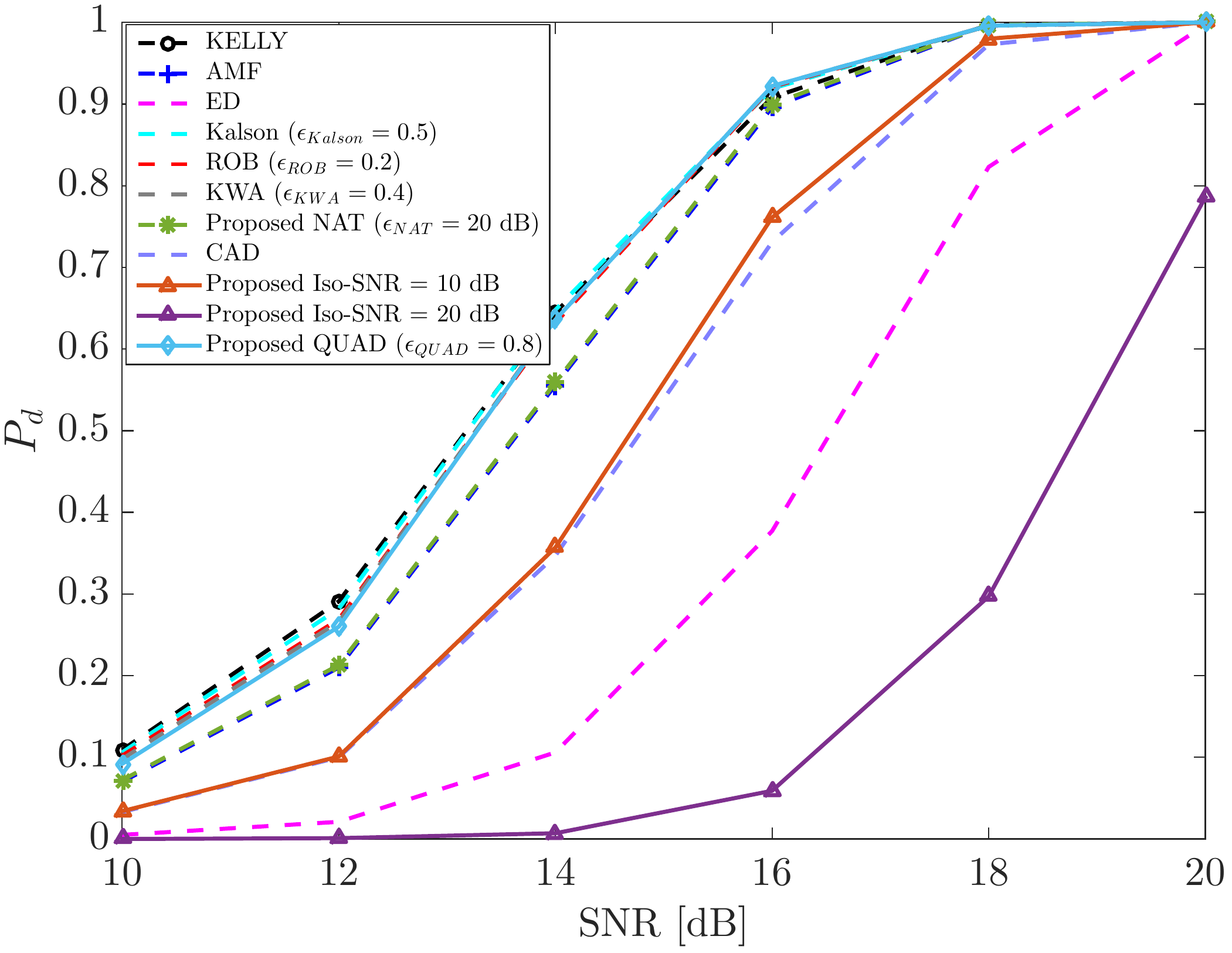}}  %
\subfigure[mismatched]{\includegraphics[width=8cm]{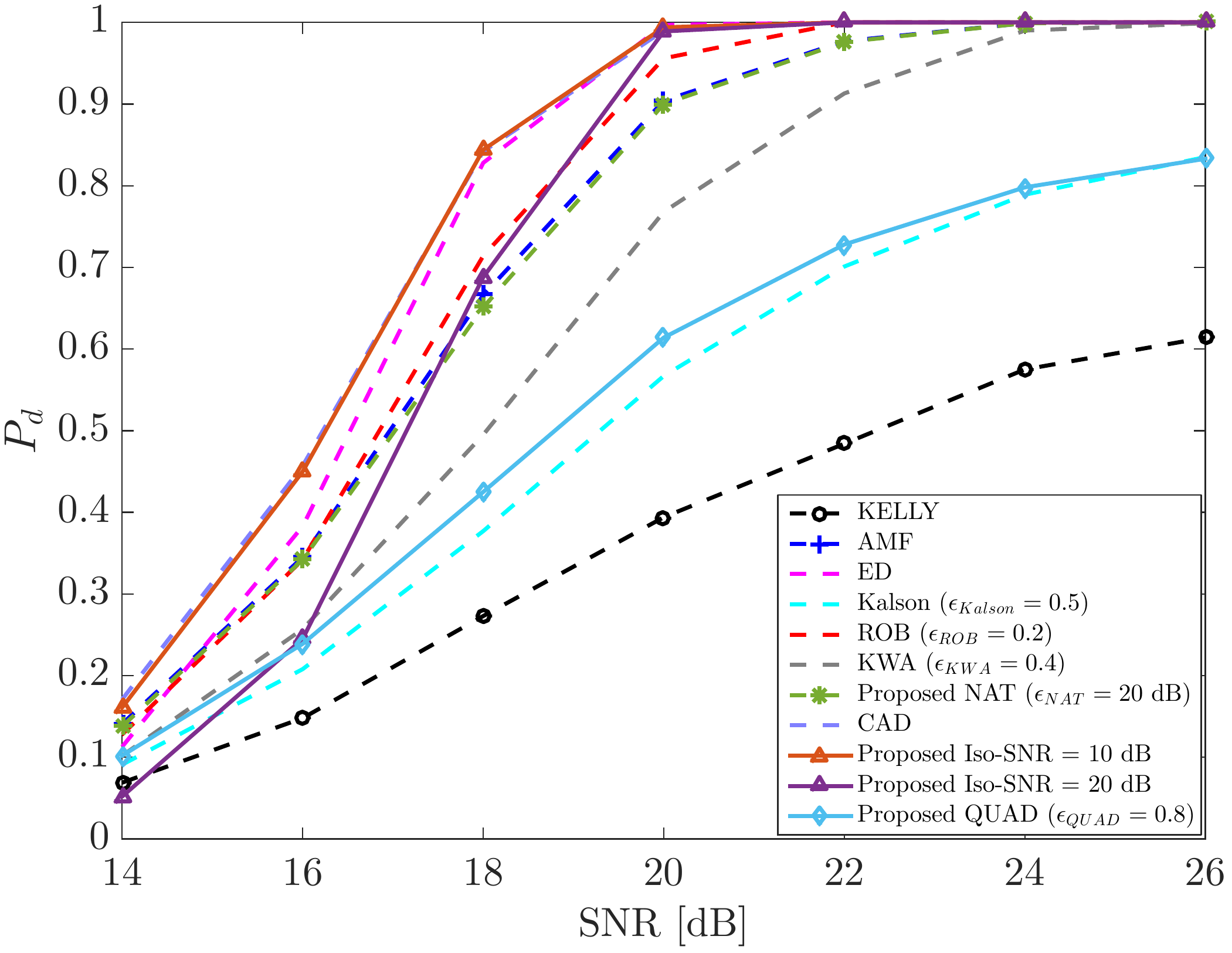}} %
\caption{Comparison among robust receivers, including the proposed detectors NAT (for $\epsilon_\text{\tiny NAT}=20$ dB) and iso-SNR for two values of the reference SNR used in the design.}\label{fig:robust}
\end{figure}
\begin{figure}
\centering
\subfigure[matched]{\includegraphics[width=8cm]{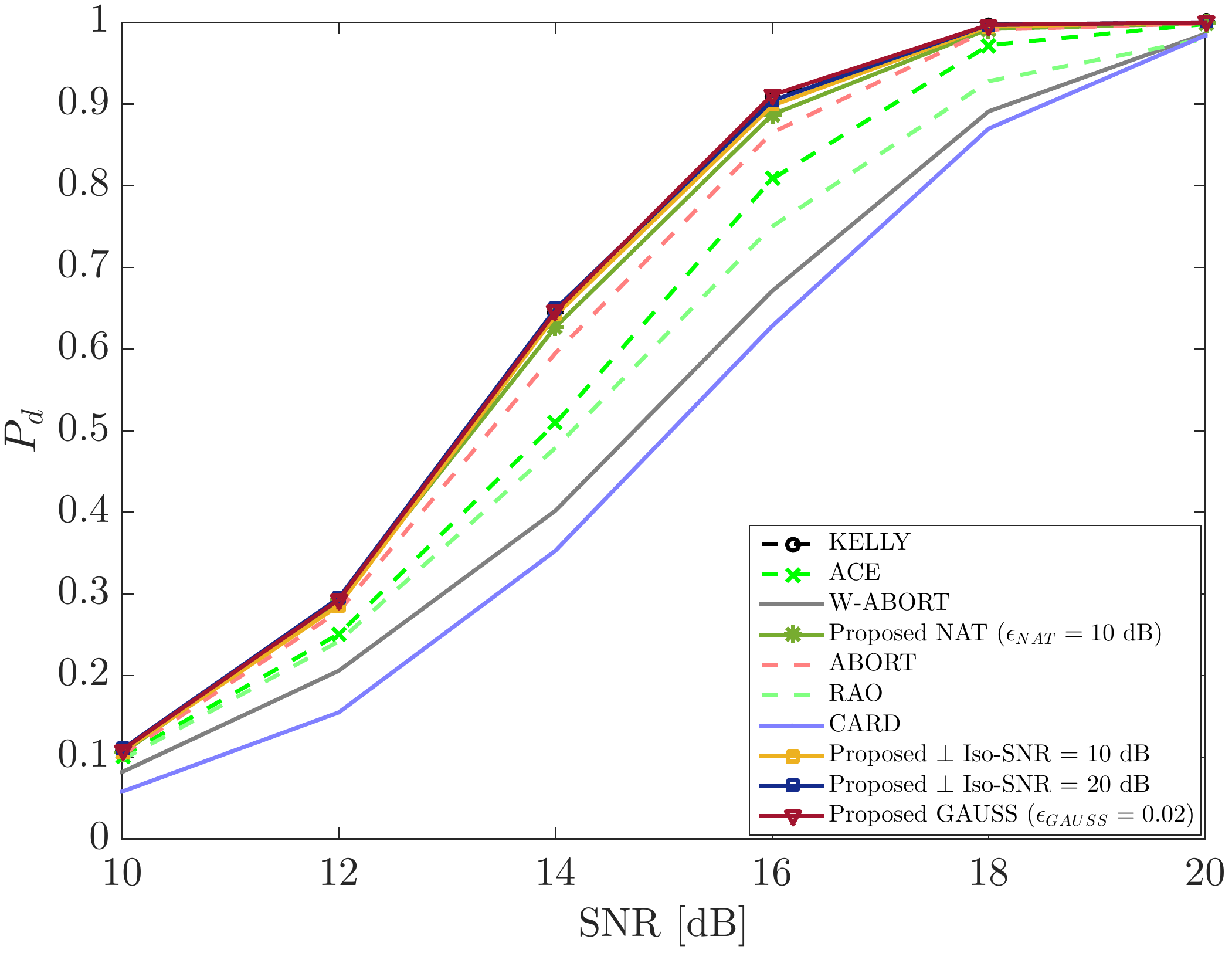}}  %
\subfigure[mismatched]{\includegraphics[width=8cm]{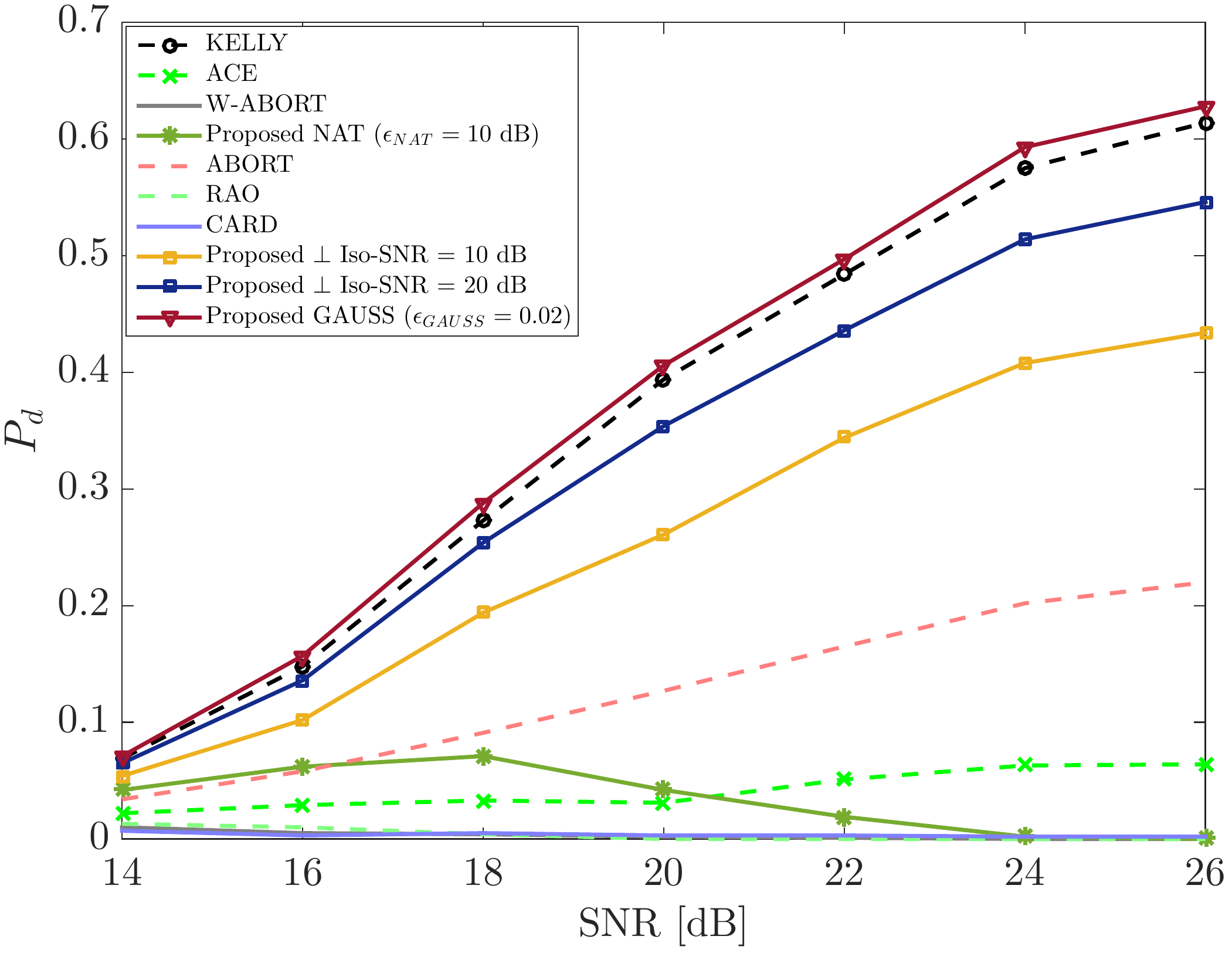}}  %
\caption{Comparison among selective receivers including the proposed detectors NAT (for $\epsilon_\text{\tiny NAT}=10$ dB) and $\perp$iso-SNR for two values of the reference SNR used in the design.}\label{fig:selective}
\end{figure}

We assess the performance by Monte Carlo simulation with $100/P_{fa}$ independent trials to set the thresholds, $P_{fa}=10^{-4}$,
and $10^3$ independent trials to compute the $P_ d$s. We consider $N=16$ and $K=32$, all the other parameters are set as in Sec. \ref{sec:interpretation}.
Results under (a) matched and (b) mismatched conditions are reported as subfigures in Figs. \ref{fig:robust}-\ref{fig:selective}, for robust and selective receivers, respectively.

Regarding detectors with robust behavior, what is interesting is that it is possible to obtain practically the same performance of some well-known receivers although the latter were obtained by a radically different design approach (typically a GLRT) compared to the CFAR-FP framework proposed here. For instance from Fig. \ref{fig:robust} it is apparent that the performance of NAT with $\epsilon_\text{\tiny NAT}=20$ dB is practically identical to AMF (under both matched and mismatched conditions); similarly, the Iso-SNR ($10$ dB) behaves pretty much as the CAD, i.e., compared to Kelly's detector it is more robust but experiences a larger loss under matched conditions (yet less than $2$ dB, while e.g. the loss of ED is $5$ dB). In Fig. \ref{fig:robust}(a) it is also apparent that the proposed QUAD has the same $P_d$ of Kelly's detector (except for a very small loss for lower SNRs) while it is more robust, with performance under mismatched conditions that are in between AMF and Kelly's detector, somewhat close to the behavior of Kalson's detector. These results are also in perfect agreement with the analysis of the decision region boundary in the CFAR-FP performed earlier; in particular by comparing Fig. \ref{fig:regioni_proposed} with Fig. \ref{fig:robustiN16K32} it is easy to realize the closeness of some of the curves, which thus imply almost identical performance (under both matched and mismatched conditions).

As regarding selective detectors, we notice that the proposed GAUSS performs very close to Kelly's detector. A remarkable result is instead achieved by the proposed NAT with  $\epsilon_\text{\tiny NAT}=10$ dB, which guarantees the same $P_d$ of Kelly's detector (under matched conditions) but at the same time very strong selectivity. This behavior has no counterpart in any existing detector, since a more or less significant loss under matched conditions (compared to Kelly's detector) is observed in all well-known selective receivers, as visible in Fig. \ref{fig:selective}(b) by comparison with ACE, ABORT/W-ABORT, CARD and RAO.
From Fig. \ref{fig:selective}(a) we see that also the $\perp$iso-SNR (with both considered values of SNR) has practically no loss under matched conditions with the same or a bit stronger selectivity (compared to Kelly's detector), ref. Fig. \ref{fig:selective}(b). Indeed, by looking again at Fig. \ref{fig:regioni_proposed}, it can be immediately visualized that the decision region boundary of the NAT with  $\epsilon_\text{\tiny NAT}=10$ dB is rather peculiar with respect to the other detectors, and apparently more suitable to isolate the matched $H_1$ cluster (with $\cos^2\theta=1$) from all the other clusters ($H_0$ but also $H_1$ under mismatches) while achieving the same $P_d$ under matched conditions probably due to the fact that it turns almost horizontal to avoid including too much $H_0$ points in the upper-most side of the boundary.  

\section{Conclusion}\label{sec:conclusions}

This paper addressed the analysis and design of radar detectors that can be represented in 
the CFAR feature plane given by the maximal invariant $(\beta,\tilde{t})$.
To this end, data clusters have been analytically described in such a plane  in terms of trajectories and shapes  under both $H_0$ and matched/mismatched steering vector  $H_1$ hypotheses. Obtained results allow one to understand the performance   of a detector (under matched conditions, but also its robustness or selectivity to mismatches) starting from the form of its decision region boundary. 

This new framework not only sheds new light on the properties of existing detectors, but also enables the design of novel linear and non-linear detectors. Several criteria have been provided to this aim, but clearly much more are possible. 
The performance assessment already showed however that through the CFAR-FP framework it is possible to achieve the same performance of  well-known receivers (obtained by a radically different design approach, typically a GLRT) as well as remarkable behaviors that have no counterpart in classical detectors; in particular, for the first time a selective detector has been obtained which exhibits no $P_d$ loss under matched conditions and, at the same time, the very strong rejection capabilities of state-of-the-art selective detectors (namely ABORT/W-ABORT) that, conversely, experience a significant loss  under matched conditions. 

As future work we will continue to investigate the design (and analysis) of detectors through the proposed framework, which by design  always guarantees (the chosen $P_{fa}$ and) the CFAR property, with the aim of further exploring the palette of achievable robust or selective behaviors while keeping under control the performance under matched conditions.

\appendices

\section{Proof of Proposition \ref{teo:cluster}}\label{app:A}

We recall the general characterization of $(\beta,\tilde{t})$ parameterized in $\gamma$ and $\cos^2 \theta$, which encompasses the one under $H_0$ (for $\gamma=0$) and $H_1$ under matched conditions (for $\cos^2 \theta=1$): $\tilde{t}$ given $\beta$ is ruled by a complex noncentral F-distribution with $1$ and $K-N+1$ complex degrees of freedom and noncentrality parameter $\gamma \beta \cos^2\theta$, i.e., $\tilde{t}\sim \mathcal{CF}_{1,K-N+1}(\gamma \beta \cos^2\theta)$; $\beta$ is ruled by a complex noncentral Beta distribution with $K-N+2$ and $N-1$ complex degrees of freedom and noncentrality parameter $\gamma (1-\cos^2\theta)$, i.e., $\beta\sim \mathcal{C\beta}_{K-N+2,N-1}(\gamma (1-\cos^2\theta))$. %

Notice that: $\tilde{t}\sim \mathcal{CF}_{1,K-N+1}(\gamma \beta \cos^2\theta)$ given $\beta$ can be recast as $\frac{1}{K-N+1}$ times a real noncentral F-distribution with $2$ and $2(K-N+1)$ (real) degrees of freedom and noncentrality parameter $2\gamma \beta \cos^2\theta$;  $\beta\sim \mathcal{C\beta}_{K-N+2,N-1}(\gamma (1-\cos^2\theta))$ can be recast as a proper flipped real Beta, i.e., 1 minus a real noncentral Beta distribution with parameters $N-1$ and $K-N+2$ and noncentrality parameter $2\gamma (1-\cos^2\theta)$ \cite{handbook}.

To describe the clusters, it is sufficient to find their center and shape, which are linked to the mean and covariance matrix of the vector $\bm{\xi}=[\beta\ \tilde{t}]^T$; these can be derived by starting from the formulas of the (generally noncentral) real F and Beta distributions in \cite{handbook} and using the relationships with the complex equivalents above.
For the first moment we have
\begin{equation}
\E[\tilde{t} | \beta] = \frac{1}{K-N} (1+\gamma \beta \cos^2\theta). \label{eq:cond_mean}
\end{equation}
By the theorem of iterated mean we obtain
$$
\E[\tilde{t}] = \E[\E[\tilde{t} | \beta]] = \frac{1}{K-N} (1+\gamma \mu_\beta \cos^2\theta)
$$
where
\begin{align*}
\mu_\beta = \E[\beta] &= 1 - \e^{-\gamma \sin^2\theta} \frac{\Gamma(N)\Gamma(K+1)}{\Gamma(N-1)\Gamma(K+2)}\\
&  \times{}_2F_2(K+1, N; N-1, K+2; \gamma\sin^2\theta)
\end{align*}
which can be easily recast into $\eqref{eq:mu_beta}$, also using the property of Euler's Gamma function $\Gamma(n)=(n-1)\Gamma(n-1)=(n-1)!$.
The calculation of $\VAR[\beta]$ in eq. \eqref{eq:var_beta} proceeds similarly by using the formula for the variance 
of a real Beta distribution \cite{handbook} and the well-known identity $\VAR[\beta]=\E[\beta^2]-(\E[\beta])^2$.  

The expression of $\VAR[\tilde{t}]$ requires a more articulated procedure; by the law of total variance we have that 
\begin{align*}
\VAR[\tilde{t}] &= \E[\VAR[\tilde{t} | \beta]] + \VAR[\E[\tilde{t}|\beta]]\\
& = \frac{1}{(K-N)^2 (K-N-1)} \Big[ (\gamma\cos^2\theta)^2 \E[\beta^2] \\
&+ (1\! +\! 2\gamma\mu_\beta \cos^2\theta)(K\! -\! N\! +\! 1)\Big] + \frac{(\gamma\cos^2\theta)^2 \VAR[\beta]}{(K-N)^2}
\end{align*}
where we have used the formulas for the mean and variance of the real F-distribution; eq. \eqref{eq:var_t_tilde} is finally obtained by using  $\E[\beta^2]= \sigma_\beta^2 + \mu_\beta^2$.

The last term to be computed is the correlation coefficient $\rho$; from the well-known identity
$
\E[(\tilde{t}-\mu_{\tilde{t}} )(\beta-\mu_\beta)] = \E[\tilde{t}\beta]-\mu_{\tilde{t}} \mu_\beta
$ 
and using also eq. \eqref{eq:cond_mean}, we obtain
\begin{align*}
\E[\tilde{t}\beta] &= \E[\E[\tilde{t}\beta | \beta]] = \E\left[ \beta \frac{1}{K-N} (1+\gamma \beta \cos^2\theta) \right]\\
&=\frac{\mu_\beta}{K-N} + \frac{1}{K-N} \gamma \E[\beta^2] \cos^2\theta
\end{align*}
from which the final expression \eqref{eq:rho} is readily obtained.

\section{Proof of Proposition \ref{teo:linear_migration}}\label{app:B}

We start by noticing that  eq. \eqref{eq:mu_beta} can be rewritten as
\begin{align}
& \mu_\beta = 1 - \frac{N-1}{K+1} \e^{-x} \nonumber \\
&\times \sum_{n=0}^\infty \frac{\Gamma(K+1+n) \Gamma(N+n) \Gamma(N-1) \Gamma(K+2)}{\Gamma(K+1)\Gamma(N) \Gamma(N-1+n) \Gamma(K+2+n)} \frac{x^n}{n!} \nonumber\\
& = 1- \e^{-x} \sum_{n=0}^\infty a_n \frac{x^n}{n!} \label{eq:mu_beta_series}
\end{align}
where $x=\gamma \sin^2 \theta$ and $a_n = \frac{N-1+n}{K+1+n}$.
The series on the righthand side resembles pretty much the Taylor's expansion of $\e^x$; in particular, adding the same number $n$ to numerator and denominator will make the resultant fraction closer to 1 than the starting fraction $\frac{N-1}{K+1}\in(0,1)$, which is already typically closer to 1 than to 0 because it is desirable to keep $K$ not much larger than $N$ (a common case if $K=2N$ or a bit less, since it is not easy to obtain a large number of secondary data in real applications).
This of course applies to any term of the progression ($a_n$ is obtained by adding one to the numerator and denominator of $a_{n-1}$), making it quickly convergent to 1; thus, we can state that the sum of the series deviates from  $\e^x$ mainly because of a limited number of lower $a_n$ terms. More precisely, it is a simple matter to show that the series above, seen as a function of $x$, lies in a quite narrow stripe between $\e^x$ and a slightly rescaled version of its, i.e.
$
\frac{N-1}{K+1} \e^x < \sum_{n=0}^\infty a_n \frac{x^n}{n!} < \e^x.
$
Hence,  one expects a weak dependency of $\mu_\beta$ on $x$ (hence ultimately on  $\cos^2\theta$ as well as $\gamma$), which from  \eqref{eq:mu_t} implies that the parametric curve $(\mu_\beta(\cos^2\theta), \mu_{\tilde{t}}(\cos^2\theta))$  should not deviate too much from a straight line. 

This intuition can be made more rigorous by explicitly computing the  sum of the series. First observe that
$$
\sum_{n=0}^\infty a_n \frac{x^n}{n!} = (N-1) \sum_{n=0}^\infty \frac{x^n}{(K+1+n) n!} + \sum_{n=0}^\infty \frac{n x^n}{(K+1+n) n!}
$$
and that the expansion of the lower incomplete Gamma is
$$
\gamma(s,x) = \int_0^x t^{s-1} \e^{-t} \mathrm{d}t =  \sum_{n=0}^\infty \frac{(-1)^n x^{s+n}}{(s+n) n!}.
$$
It immediately follows that
\begin{align*}
 \sum_{n=0}^\infty \frac{x^n}{(K+1+n) n!} & = (-x)^{-(K+1)} (-1)^{-(K+1)} \\
 & \times \sum_{n=0}^\infty \frac{(-x)^{n+K+1}}{(K+1+n) n!} (-1)^{n+K+1}\\
 & = (-x)^{-(K+1)} \gamma(K+1,-x).
\end{align*}
Moreover,
\begin{align*}
 \sum_{n=0}^\infty \frac{n x^n}{(K+1+n) n!} & =  \sum_{n=1}^\infty \frac{ x^n}{(K+1+n) (n-1)!}\\
 &= \sum_{m=0}^\infty \frac{ x^{m+1}}{(K+2+m) m!}\\
 &= - (-x)^{-(K+1)}\! \sum_{m=0}^\infty \! \frac{(-1)^m (-x)^{K+2+m}}{(K+2+m) m!}\\
 & = - (-x)^{-(K+1)} \gamma(K+2,-x)
\end{align*}
where we have used the change of index $m=n-1$. We get
$$
\sum_{n=0}^\infty a_n \frac{x^n}{n!} \!=\! (-x)^{-(K + 1)} \! \big[ (N\!-\! 1) \gamma(K\! +\! 1,-x) - \gamma(K\! +\! 2,-x) \big]
$$
which, from the recurrence relation $\gamma (s+1,x)=s\gamma (s,x)-x^{s}\mathrm {e} ^{-x}$,  can be rewritten as
$$
\sum_{n=0}^\infty a_n \frac{x^n}{n!} -  \e^x  = (N-K-2) (-x)^{-(K+1)} \gamma(K+1,-x).
$$
The latter expression quantifies more explicitly the closeness between the  series and $\e^x$, and substituted in \eqref{eq:mu_beta_series} finally yields
\begin{equation}
\mu_\beta= (K-N+2) (-x)^{-(K+1)} \e^{-x} \gamma(K+1,-x). \label{eq:mu_beta_bis}
\end{equation}
Now, since the Gamma function grows very fast and $K$ is a non-small value in practice, we can consider the asymptotic expansion  $\frac{\gamma(s,z)}{\Gamma(s)} \sim \frac{z^s \e^{-z}}{\Gamma(s+1)}$ \cite[eq. 8.11.5]{NIST}, which implies  that in \eqref{eq:mu_beta_bis}
$$
\gamma(K+1,-x) \sim \frac{\Gamma(K+1)}{\Gamma(K+2)} (-x)^{K+1} \e^x = \frac{(-x)^{K+1} \e^x}{K+1}
$$ 
so basically compensating the functional variation.
We conclude that the intuition of weak dependency of $\mu_\beta$ on $x$ is well-founded from an analytical point of view, hence the trajectory of $(\mu_\beta,\mu_{\tilde{t}})$ in the $\beta$-$\tilde{t}$ plane, as a function of either $\gamma$ or $\cos^2 \theta$,  can be approximated by a line.
For instance, to find the expression for fixed $\gamma$, it is thus sufficient to write the equation of the line passing through the points $(\mu_\beta|_{\cos^2\theta=0}, \mu_{\tilde{t}}|_{\cos^2\theta=0})$ and $(\mu_\beta|_{\cos^2\theta=1}, \mu_{\tilde{t}}|_{\cos^2\theta=1})$. From eqs. \eqref{eq:mu_beta}-\eqref{eq:mu_t} we obtain
$$
\frac{\beta - 1 + g(\gamma)}{1-\frac{N-1}{K+1} - 1 + g(\gamma)} = \frac{\tilde{t} - \frac{1}{K-N}}{\frac{1}{K-N} [1+\gamma(1-\frac{N-1}{K+1})] - \frac{1}{K-N}}
$$
where $g(\gamma)$ is defined in eq. \eqref{eq:g}, then $m$ and $q$ follow straight.
It turns out that the approximation is generally very good: it can be numerically verified that e.g. for $N=16$ and $K=32$ the relative error is always below $0.5\%$ in the whole range of $\gamma$ from 0 to 25 dB.

The equations for fixed $\cos^2\theta$ can be similarly obtained.

\end{document}